\documentclass[a4paper, onecolumn, 11pt, longbibliography, accepted=2024-08-20]{quantumarticle}
\pdfoutput=1

\usepackage{silence}
\usepackage[utf8]{inputenc}
\usepackage[english]{babel}
\usepackage[T1]{fontenc}
\usepackage{hyperref}
\usepackage{subfigure,epsfig}

\usepackage{amsmath,amssymb,amsthm,mathrsfs,amsfonts,dsfont}
\usepackage[linesnumbered, ruled, vlined]{algorithm2e}
\usepackage{bbm}
\usepackage{booktabs}
\usepackage{braket}
\usepackage{bm}
\usepackage{calc}
\usepackage{csvsimple}
\usepackage{enumerate}
\usepackage{enumitem}
\usepackage{here}
\usepackage{ifthen}
\usepackage{listings}
\usepackage{mathtools}
\usepackage{multirow}
\usepackage[square,sort,comma,numbers]{natbib}
\usepackage{optidef}
\usepackage{physics}
\usepackage{qcircuit}
\usepackage{siunitx}
\usepackage{stfloats}
\usepackage{subcaption}
\usepackage{subfiles}
\usepackage{tikz}
\usepackage{thm-restate}
\usepackage{url}
\usepackage{xparse}

\newcommand{\black}[1]{\textcolor{black}{#1}}

\newtheorem{theorem}{Theorem}

\newtheorem{lemma}{Lemma}

\newtheorem{corollary}{Corollary}

\newcommand{\Amat}{\bm{A}}
\newcommand{\Stab}{\mathrm{Stab}}
\newcommand{\STAB}{\mathrm{STAB}}

\newcommand{\Fp}[1]{\mathbb{F}_{#1}}
\newcommand{\qBinom}[2]{\genfrac{[}{]}{0pt}{}{#1}{#2}}

\newcommand{\stnorm}[1]{\norm{#1}_{\mathrm{st}}}

\newcommand{\Cpp}{C\nolinebreak[4]\hspace{-.05em}\raisebox{.4ex}{\relsize{-2}{\textbf{++}}}}
\newcommand{\defeq}{\vcentcolon=}

\newcommand{\relmiddle}[1]{\mathrel{}\middle#1\mathrel{}}

\ExplSyntaxOn
\NewDocumentCommand{\definealphabet}{mmmm}{
\int_step_inline:nnn{`#3}{`#4}{
\cs_new_protected:cpx{#1 \char_generate:nn{##1}{11}}{
\exp_not:N #2{\char_generate:nn{##1}{11}}}}}
\ExplSyntaxOff

\definealphabet{bb}{\mathbb}{A}{Z}
\definealphabet{rm}{\mathrm}{A}{Z}
\definealphabet{cal}{\mathcal}{A}{Z}
\definealphabet{frak}{\mathfrak}{a}{z}

\usepackage{hyperref}
\hypersetup{colorlinks=true,linkcolor=blue,citecolor=blue,urlcolor=blue}

\SetKwComment{Comment}{/* }{ */}

\DontPrintSemicolon

\graphicspath{{./fig/}}  
\begin{document}
\title{Handbook for Quantifying Robustness of Magic}

\author{Hiroki Hamaguchi}
\email{hamaguchi-hiroki0510@g.ecc.u-tokyo.ac.jp}
\affiliation{Graduate School of Information Science and Technology, University of Tokyo, Tokyo, 7-3-1 Hongo, Bunkyo-ku, Tokyo 113-8656, Japan}

\author{Kou Hamada}
\email{zkouaaa@g.ecc.u-tokyo.ac.jp}
\affiliation{Graduate School of Information Science and Technology, University of Tokyo, Tokyo, 7-3-1 Hongo, Bunkyo-ku, Tokyo 113-8656, Japan}

\author{Nobuyuki Yoshioka}
\email{nyoshioka@ap.t.u-tokyo.ac.jp}
\affiliation{Department of Applied Physics, University of Tokyo, 7-3-1 Hongo, Bunkyo-ku, Tokyo 113-8656, Japan}
\affiliation{Theoretical Quantum Physics Laboratory, RIKEN Cluster for Pioneering Research (CPR), Wako-shi, Saitama 351-0198, Japan}
\affiliation{JST, PRESTO, 4-1-8 Honcho, Kawaguchi, Saitama, 332-0012, Japan}


\begin{abstract}
Nonstabilizerness, or magic, is an essential quantum resource for performing universal quantum computation. Specifically, Robustness of Magic (RoM) characterizes the usefulness of a given quantum state for non-Clifford operations. While the mathematical formalism of RoM is concise, determining RoM in practice is highly challenging due to the necessity of dealing with a super-exponential number of stabilizer states.
In this work, we present novel algorithms to compute RoM. The key technique is a subroutine for calculating overlaps between all stabilizer states and a target state, with the following remarkable features: (i) the time complexity per state is reduced \textit{exponentially}, and (ii) the total space complexity is reduced \textit{super-exponentially}. Based on this subroutine, we present algorithms to compute RoM for arbitrary states up to $n = 8$ qubits, whereas the naive method requires a memory size of at least $\SI{86}{\pebi\byte}$, which is infeasible for any current classical computer.
Additionally, the proposed subroutine enables the computation of stabilizer fidelity for a mixed state up to $n = 8$ qubits. We further propose novel algorithms that exploit prior knowledge of the structure of the target quantum state, such as permutation symmetry or disentanglement, and numerically demonstrate our results for copies of magic states and partially disentangled quantum states.
This series of algorithms constitutes a comprehensive ``handbook'' for scaling up the computation of RoM. We envision that the proposed technique will also apply to the computation of other quantum resource measures.
\end{abstract}
\maketitle

\section{Introduction}\label{sec:intro}

Universal fault-tolerant quantum computation is often formulated such that the elementary gates consist of both classically simulatable gates and costful gates, such as in the most well-known Clifford+$T$ formalism of the magic state model~\cite{gottesmanHeisenbergRepresentationQuantum1998, nielsenQuantumComputationQuantum2010, bravyiUniversalQuantumComputation2005, litinskiGameSurfaceCodes2019, horsmanSurfaceCodeQuantum2012, fowlerLowOverheadQuantum2019,
Veitch_2012}.
Since the consumption of the non-Clifford gates is indispensable for any quantum advantage~\cite{gidneyHowFactor20482021,leeEvenMoreEfficient2021, vonburgQuantumComputingEnhanced2021, yoshiokaHuntingQuantumclassicalCrossover2023}, there is a surging need to evaluate the complexity of quantum circuits using the framework of resource theory, in order to explore the boundary of quantum and classical computers~\cite{bravyiTradingClassicalQuantum2016, tirritoQuantifyingNonstabilizernessEntanglement2023, hahnQuantifyingQubitMagic2022, haugEfficientStabilizerEntropies2023, seddonQuantifyingQuantumSpeedups2021, liuManybodyQuantumMagic2022, leoneStabilizerEnyiEntropy2022}.
One such attempt is the proposal of a quantity called Robustness of Magic (RoM) by Howard and Campbell~\cite{howardApplicationResourceTheory2017a,
pashayanEstimatingOutcomeProbabilities2015,
heinrichRobustnessMagicSymmetries2019}; RoM characterizes the classical simulation overhead or complexity of a given quantum state based on its effective amount of magic, and it geometrically quantifies the distance from the convex set of stabilizer states.
\black{More advanced resource measures have been proposed for quantifying the classical simulation cost, such as the stabilizer rank/extent~\cite{bravyiTradingClassicalQuantum2016,bravyiSimulationQuantumCircuits2019} and dyadic negativity~\cite{seddonQuantifyingQuantumSpeedups2021}, while RoM remains one of the most insightful measures in the context of gate synthesis. Together with the fact that the mathematical formulation of RoM is among the simplest and, hence, insightful for building large-scale computation techniques for other monotones as well, we focus on the computation of RoM in this work.}

Although we can reduce the computation of RoM to a simple $L^1$ norm minimization problem, since the set of stabilizer states grows super-exponentially with the number of qubits $n$, the time and space complexity become prohibitively large, making computations for $n > 5$ qubits extremely challenging. For instance, memory consumption explodes to $\SI{86}{\pebi\byte}$ already for an $n = 8$ qubit system. Existing works have attempted to mitigate this burden; for example, Heinrich et al.~proposed leveraging the symmetry of the target state~\cite{heinrichRobustnessMagicSymmetries2019}. By exploiting the permutation symmetry between copies of identical states and internal (or local) symmetry, they demonstrated that up to $n = 26$ qubits RoM can be computed for symmetric magic states such as $\ket{H}^{\otimes n}$ used for $T$-gates. However, when investigating the magic resource of noisy states, there has been no valid method to scale up the resource characterization.

In this work, we propose a systematic procedure, presented in Table~\ref{table:solutions} and Fig.~\ref{fig:flow_chart}, to compute RoM that overcomes the limitations of existing methods. Central to our approach is a subroutine that computes the overlaps of a given quantum state with all stabilizer states, featuring (i) exponentially faster time complexity per state and (ii) super-exponentially smaller space complexity in total. Utilizing this subroutine, we present algorithms that surpass the current state-of-the-art RoM calculation results for arbitrary states up to $n = 8$ qubits. Furthermore, we extend the capability of these methods by incorporating preknowledge of the target quantum state's structure, such as permutation symmetry and decoupled structure, demonstrating that we can approximate RoM for multiple copies of arbitrary single-qubit states up to $n = 17$ qubits.

The remainder of this work is organized as follows.
In Sec.~\ref{sec:preliminary}, we present the preliminaries regarding the formalism of RoM.
In Sec.~\ref{sec:scale_up}, we first give the main subroutine on the overlap calculation in Theorem~\ref{thm:inner_product} and then present the algorithms that compute RoM with the reduced computational resource by utilizing the information of overlaps.
In Sec.~\ref{sec:application}, we present algorithms for practical target states that are decoupled from each other, such as the multiple copies of single-qubit states or tensor products over subsystems.
Finally, in Sec.~\ref{sec:discussion}, we provide the discussion and future perspective of our work.

\begin{table*}[b]
    \centering
    \begin{tabular}{|c|c|c|c|}
    \hline 
    Method & Target & Qubit count & Exact/Approximate \\ \hline \hline
    Naive LP~\cite{howardApplicationResourceTheory2017a}  & Arbitrary & $n \leq 5$  & Exact \\
    Column Generation (CG)    & Arbitrary & $n \leq 8$  & Exact\\
    Minimal Feasible Solution    & Arbitrary & $n \leq 14$  & $2^n$-approximation\\
    \hline
    Symmetry Reduction & $\rho^{\otimes n}$ & $n \leq 17$ & Exact up to $n \leq 7$ \\ 
    Partition Optimization & $\bigotimes_{i} \rho_i$ & $n \leq 15$ & Approximation                \\
    Symmetry Reduction~\cite{heinrichRobustnessMagicSymmetries2019}& $\rho_{H, F}^{\otimes n}$ & $n \leq 26$ & Exact up to $n \leq 9, 10$ \\ \hline
    \end{tabular}
    \caption{Methods for calculating RoM. We can apply the top three methods to arbitrary $n$-qubit states, whereas the subsequent three assume specific structures such as permutation symmetry, decoupled structure, and local symmetry, respectively.}
    \label{table:solutions}
\end{table*}
\begin{figure}[tb]
    \centering
    \includegraphics[width=\columnwidth]{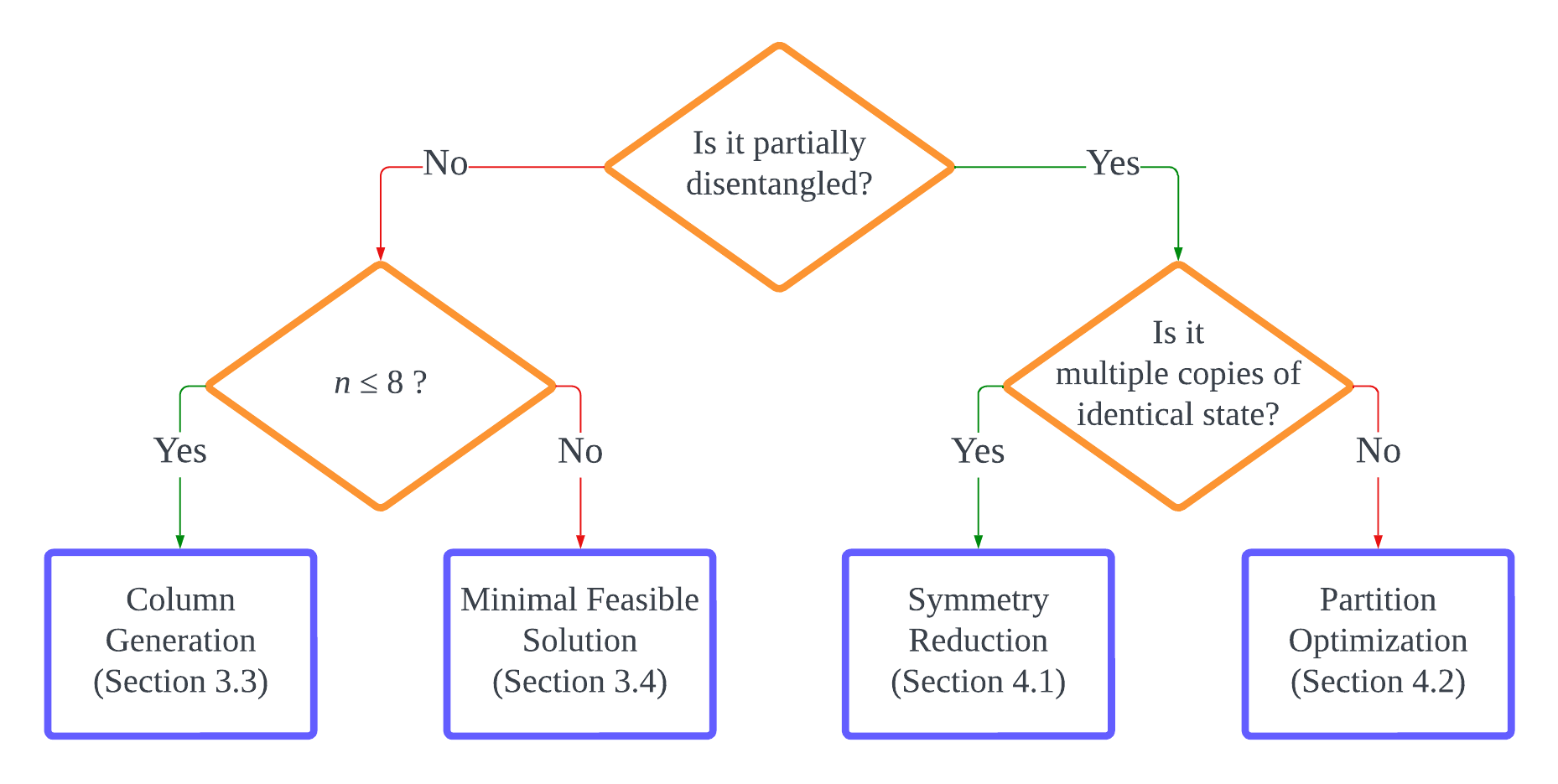}
    \caption{Flow chart of RoM computation. }
    \label{fig:flow_chart}
\end{figure}

\section{Preliminaries}\label{sec:preliminary}

\subsection{Stabilizer State}
Let $\calP_n = \{\pm1, \pm i\} \times \{I, X, Y, Z\}^{\otimes n}$ denote the $n$-qubit Pauli group. For any $n$-qubit stabilizer state $\ket{\phi}$, we define the stabilizer group of $\ket{\phi}$ as $\Stab(\ket{\phi}) = \langle P_1, \dots, P_n \rangle$, where $P \in \Stab(\ket{\phi})$ satisfies $P \ket{\phi} = \ket{\phi}$ and each $P_i \in \calP_n$ is an independent stabilizer generator.
We denote the entire set of $n$-qubit stabilizer states by $\calS_n$, whose size scales super-exponentially as $|\calS_n| = 2^n \prod_{k=0}^{n-1} (2^{n-k} + 1) = 2^{\order{n^2}}$~\cite[Proposition 2]{aaronsonImprovedSimulationStabilizer2004}\cite[Proposition 2]{garciaGeometryStabilizerStates2017}. Additionally, we denote their convex hull as $\STAB_n = \{\sum_j x_j \sigma_j \mid \sigma_j \in \calS_n, x_j \geq 0, \sum_j x_j = 1\}$.
\black{There are concise representations of a stabilizer state, such as the stabilizer tableau~\cite{aaronsonImprovedSimulationStabilizer2004} or the check matrix representation~\cite[Section 10.5.1]{nielsenQuantumComputationQuantum2010}. In this paper, we use the check matrix representation.
For a stabilizer group $\langle P_1, \dots, P_n \rangle$, let each generator $P_i$ be expressed as
\begin{equation}\label{eq:binaryFormOfPi}
    P_i = (-1)^{\delta_i}P_{i,1} \otimes P_{i,2} \otimes \dots  \otimes P_{i,n},
\end{equation}
where $P_{i,j}\in \{I,X,Y,Z\}$ and $\delta_i \in \{0,1\}$.
Note that the complex term is unnecessary. For any element $P\in \{\pm1\} \times \{I,X,Y,Z\}^{\otimes n}$, if $iP$ is in the stabilizer group, then $(iP)(iP)=-I^{\otimes n}$ should also be an element, which is a contradiction.
Let $\bbF_2$ be the finite field with two elements.
The check matrix of the stabilizer group is an $\bbF_2^{n\times 2n}$ matrix defined as
\begin{equation*}
    C \defeq [X_n~Z_n]
\end{equation*}
where $(X_n)_{i,j}\in \bbF_2$ is 1 iff $P_{i,j}$ is $X$ or $Y$, and $(Z_n)_{i,j}\in \bbF_2$ is 1 iff $P_{i,j}$ is $Z$ or $Y$.
The $i$-th generator $P_i$ corresponds to the $i$-th row vector of the matrix $C$, which is denoted by $r(P_i)$.
Using this check matrix representation, we can confirm the independence and the commutativity of stabilizer generators as follows:
\begin{lemma}[{\cite[Proposition 10.3]{nielsenQuantumComputationQuantum2010}}]
\label{lemma:check_matrix_independent}
    The generators $\{ P_1, \dots, P_n \} (\not\ni -I^{\otimes{n}})$ are mutually independent iff the rows of the corresponding check matrix $C$ are linearly independent.
\end{lemma}
\begin{lemma}[{\cite[Exercise 10.33]{nielsenQuantumComputationQuantum2010}}]
\label{lemma:check_matrix_commutativity}
    The generators $\{ P_1, \dots, P_n \}$ are commute $([P_i,P_j]=0)$ iff the corresponding check matrix $C$ satisfies
    $C \mqty(
        0 & I_n \\
        I_n & 0
    ) C^\top = 0$, where $I_n$ is the identity matrix of size $n$.
\end{lemma}
}

\subsection{Robustness of Magic}
Robustness of Magic (RoM) of a given $n$-qubit quantum state $\rho$ can be interpreted as distance from the polytope of free states $\STAB_n$ and is defined as~\cite{howardApplicationResourceTheory2017a}\cite[Section 2]{seddonQuantifyingMagicMultiqubit2019}
\begin{equation}\label{eq:RoMDefOrig}
    \calR(\rho) \defeq \min_{\sigma_+, \sigma_-\in\STAB_n}\left\{2p+1 \relmiddle| \rho = (p+1) \sigma_+ - p \sigma_-,~ p\geq 0\right\}.
\end{equation}
Given the entire set of stabilizer states $\calS_n=\{\sigma_j\}_{j=1}^{\abs{\calS_n}}$, it is straightforward to demonstrate that this yields an equivalent expression as
\begin{equation*}
    \calR(\rho) = \min_{\bm{x}\in\bbR^{\abs{\calS_n}}}\left\{\norm{\bm{x}}_1 \relmiddle| \rho = \sum_{j=1}^{\abs{\calS_n}} x_j \sigma_j \right\},
\end{equation*}
which can be further simplified as
\begin{equation}\label{eq:primal_RoM}
\calR(\rho) = \min_{\bm{x}\in\bbR^{\abs{\calS_n}}}\left\{\norm{\bm{x}}_1 \relmiddle| \Amat_n \bm{x} = \bm{b} \right\}.
\end{equation}
Here, we have utilized the unique decomposition of the quantum state into $n$-qubit Pauli operators defined by $b_i = \Tr[P_i \rho]$ and $(\Amat_n)_{i,j} = \Tr[P_i \sigma_j]$ where $P_i~(1\leq i\leq 4^n)$ is the $i$-th Pauli operator in $\{I, X, Y, Z\}^{\otimes n}$. We call $\bm{b}$ as the Pauli vector of $\rho$, which can be computed with time complexity of $\order{n4^n}$ as described in Appendix~\ref{subapp:pauli_vec_innerproduct}.
\black{Let $\sigma_j \in \calS_n$ be expressed as $\ketbra{\phi_j}{\phi_j}$, and $\bm{a}_j$ denote the $j$-th column of $\Amat_n$. Note that $(\bm{a}_j)_i=(\Amat_n)_{i,j}$ is $1(=\Tr[P_i \ketbra{\phi_j}{\phi_j}])$ if $P_i$ is a stabilizer of $\ket{\phi_j}$, $-1$ if $-P_i$ is a stabilizer of $\ket{\phi_j}$, and $0$ otherwise.} 
In practice, one may solve Eq.~\eqref{eq:primal_RoM} as 
\begin{mini}
    {\bm{u}}{\sum_{i}u_i}{\label{eq:RoM_LP}}{}
    \addConstraint{\mqty( \Amat_{n} & -\Amat_{n} ) \bm{u}}{=\bm{b}}
    \addConstraint{\bm{u}}{\geq \bm{0}},
\end{mini}
where the inequality for $\bm{u}$ is element-wise. This is a standard form of Linear Programming (LP), and we can solve this by LP solvers such as Gurobi~\cite{gurobi} or
CVXPY~\cite{10.5555/2946645.3007036,agrawal2018rewriting}.

\subsection{Dualized Robustness of Magic}
\label{subsec:dualizedRoM}
Since RoM can be formalized via the standard form of the linear program, the strong duality holds, which implies that the dual problem gives an equivalent definition.
Concretely, RoM can be computed via the following:
\begin{equation}\label{eq:RoM_dual}
    \calR(\rho) = \max_{\bm{y}\in\bbR^{4^n}} \left\{ \bm{b}^{\top} \bm{y} \relmiddle| -\bm{1} \leq \Amat_n^\top \bm{y} \leq \bm{1} \right\},
\end{equation}
where $\bm{1}$ is a length-$\abs{\calS_n}$ vector with all the elements given by unity.
By nature of the dual problem, any feasible solution yields a lower bound for RoM.
For instance, RoM can be lower-bounded by the st-norm $\stnorm{\rho}=\frac{1}{2^n} \norm{\bm{b}}_1$ by taking $\bm{y}$ as $y_i = \mathrm{sgn}(b_i)/2^n$~\cite[Appendix B]{heinrichRobustnessMagicSymmetries2019}.

For the subsequent discussion, we introduce some notations. Let $\bm{C}$ be a submatrix of $\Amat_n$ with $4^n$ rows.
We define the set of all columns in $\Amat_n$ as $\calA_n =\{\bm{a}_j\}_{j=1}^{\abs{\calS_n}}$,
and that in $\bm{C}$ as $\calC \subseteq \calA_n$.
We define the minimization problem $\texttt{Prob}(\bm{C}, \bm{b})$ as Eq.~\eqref{eq:RoM_LP} using the matrix $\bm{C}$ instead of $\bm{A}_n$.
We also define a function $\texttt{SolveLP}(\calC, \bm{b})$ which solves $\texttt{Prob}(\bm{C}, \bm{b})$.
It returns the primal solution $\bm{x}$ for Eq.~\eqref{eq:primal_RoM} and the dual solution $\bm{y}$ for Eq.~\eqref{eq:RoM_dual}.

\section{Scaling Up RoM Calculation for Arbitrary States} \label{sec:scale_up}

It is well known that linear programming problems are solvable in polynomial time to the matrix size.
However, the matrix size of $\Amat_n$ itself is $4^n \times |\calS_n|$ where $\abs{\calS_n}=2^{\order{n^2}}$, and hence it is impractical to use the entire $\Amat_n$ to tackle $n>5$ qubit systems~\cite{haugEfficientStabilizerEntropies2023}.

Motivated by this challenging scenario, we propose numerical algorithms to compute RoM for arbitrary quantum states, surpassing the state-of-the-art system size. The key technique involves overlap calculation with two notable features: (i) the time complexity is reduced from $\order{2^n\abs{\calS_n}}$ to $\order{n\abs{\calS_n}}$, representing an exponential reduction from $\order{2^n}$ to $\order{n}$ per stabilizer state, and (ii) the space complexity is super-exponentially reduced from $\order{2^n\abs{\calS_n}}$ to $\order{2^n}$, as we avoid explicitly constructing the entire $\Amat_n$.

Based on this subroutine, Algorithm~\ref{alg:main}, referred to as the top-overlap method, solves the primal problem~\eqref{eq:primal_RoM} using a limited set of stabilizers. These stabilizers are selected based on the largest or smallest overlaps with the target quantum states.
Algorithm~\ref{alg:column_generation}, using the Column Generation method, improves upon Algorithm~\ref{alg:main}. It iteratively adds stabilizer states to the decomposition until all the inequality constraints in the dual problem~\eqref{eq:RoM_dual} are satisfied. This ensures that the algorithm provides exact RoM.

In the following, we first present the fast overlap computation algorithm in Sec.~\ref{subsec:overlap}, and then proceed to introduce two algorithms in Sec.~\ref{subsec:top_K} and~\ref{subsec:column_generation}, respectively.
We demonstrate that these algorithms allow us to compute exact RoM for systems up to $n=8$ qubits. In this context, the memory consumption for the subroutine is reduced by a factor of $10^8$; compared to the entire $\Amat_n$ size of $\SI{86}{\pebi\byte}$, we can execute the subroutine only with $\SI{512}{\mebi\byte}$.

\subsection{Core Subroutine: Fast Computation of Stabilizer Overlaps}\label{subsec:overlap}

First, we introduce the core subroutine in our work that computes the overlaps between all the stabilizer states and the target state, i.e., $\Amat_n^\top \bm{b}$, or the stabilizer overlaps in short.
When we resort to a naive calculation, 
it requires the time complexity of  $\order{2^{n}\abs{\calS_n}}$ to compute all the stabilizer overlaps via the matrix-vector product of $\Amat_n^\top \bm{b}$, even if we utilize the sparsity of $\Amat_n$. We can speed up this subroutine to $\order{n\abs{\calS_n}}$.
\begin{theorem}[Computing stabilizer overlaps]\label{thm:inner_product}
    $\Amat_n^\top \bm{b}$ can be computed \black{in $\order{n\abs{\calS_n}}$ time and in $\order{2^n}$ extra space excluding inputs and outputs}. 
\end{theorem}
As we detail later in Sec.~\ref{subsec:top_K}, this technique is crucial for scaling up RoM calculation to larger systems.
To prove Theorem~\ref{thm:inner_product}, 
\black{it is beneficial to utilize} the Fast Walsh--Hadamard Transform (FWHT) algorithm~\cite{finoUnifiedMatrixTreatment1976}, which performs matrix-vector product operations when the matrix has a certain tensor product structure.
Here we use the unnormalized Walsh--Hadamard matrix $H_n \defeq \mqty(
    1   & 1        \\
    1   & -1
    )^{\otimes n}$,
and refer to the matrix-vector multiplication of $H_n$ as the FWHT algorithm. As evident from the pseudocode in Appendix~\ref{app:fwht_pseudocode}, its computational cost is as follows.
\begin{lemma}[Complexity of FWHT algorithm]\label{lem:Walsh--Hadamard}
    In-place matrix-vector multiplication of $H_n$ can be done with time complexity of $\order{n2^n}$ and space complexity of $\order{2^n}$.
\end{lemma}

Indeed, $\Amat_n$ is essentially constructed from unnormalized Walsh--Hadamard matrices (see Fig.~\ref{fig:Amat_2}). 
We define $\calW_n$ as the set of all sparsified Walsh--Hadamard matrices, which can be represented as $\begin{bmatrix} H_n\\ 0 \end{bmatrix} \in \bbR^{4^n \times 2^n}$ by appropriately reordering and flipping the signs of the rows.
$\Amat_n$ is formed by concatenating these $W \in \calW_n$ as follows:
\begin{figure}[tb]
    \centering
    \includegraphics[width=\linewidth]{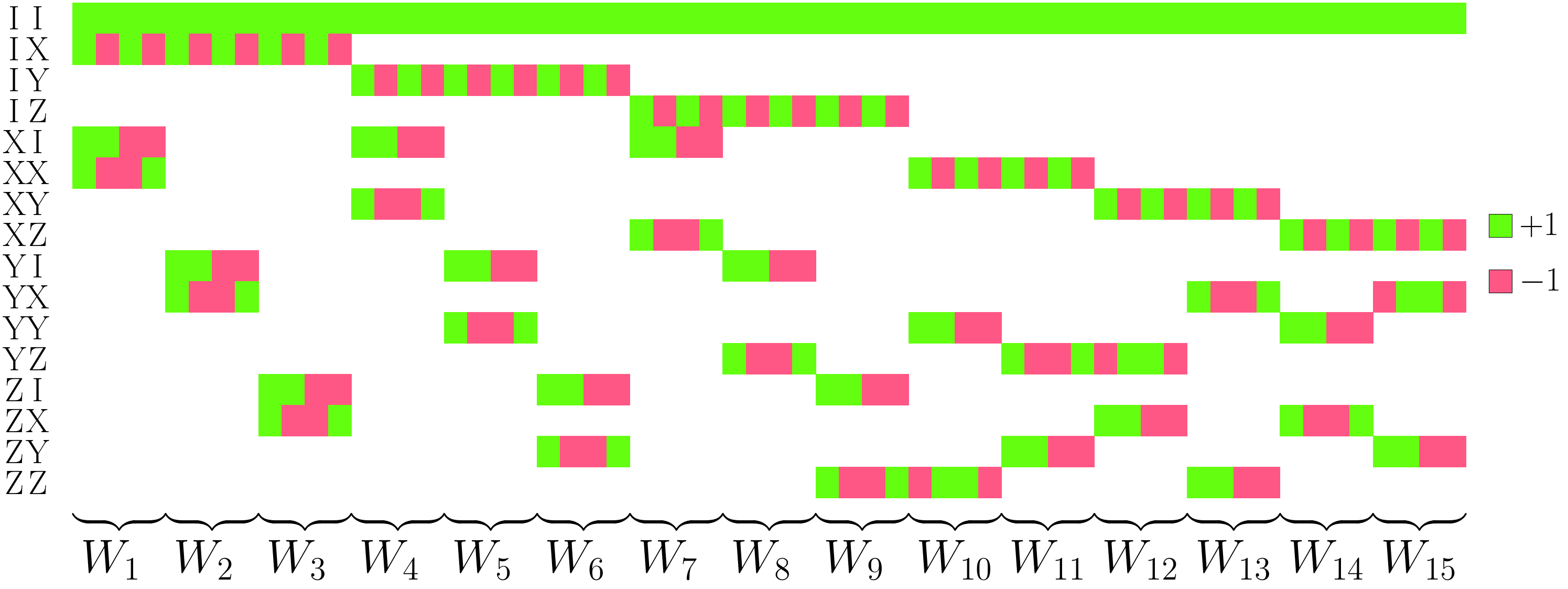}
    \caption{Visualization of $\Amat_n$ for $n=2$.
    As stated in Lemma~\ref{lem:Amat_decomposition},
    $\Amat_n$ is composed of $W_i \in \calW_n$.
    }
    \label{fig:Amat_2}
\end{figure}
\begin{lemma}[Construction of $\Amat_n$]\label{lem:Amat_decomposition}
    For all $n\in \bbN$, using $W_j\in\calW_n$, $\Amat_n$ can be written as
    \begin{equation*}
         \Amat_n = \mqty[W_1& \cdots&W_{\abs{\calS_n}/2^n}],
    \end{equation*}
    and there exists a constructive method to enumerate $\{W_j\}_{j=1}^{\abs{\calS_n}/2^n}$.
\end{lemma}
\begin{proof}
    \black{Firstly, for a check matrix $C$ that satisfies both conditions in Lemma~\ref{lemma:check_matrix_independent} and Lemma~\ref{lemma:check_matrix_commutativity}, we consider the stabilizer group $\langle P_1,\dots,P_n \rangle$ with $\delta_i=0 \; (1 \leq i \leq n)$ in Eq.~\eqref{eq:binaryFormOfPi}. We define $\delta$ as $\sum_{i=1}^{n} 2^{i-1} \delta_i$ and $\gamma$ as $\sum_{i=1}^{n} 2^{i-1} \gamma_i \; (\gamma_i \in \{0,1\})$. Each element of the stabilizer group is expressed as $P_\gamma = P_1^{\gamma_1} \cdots P_n^{\gamma_n}$. Note that all generators commute, and all elements are distinct due to the conditions of the lemmas. 
    For a stabilizer group with $\delta$, i.e., $\langle (-1)^{\delta_1} P_1, \dots, (-1)^{\delta_n} P_n \rangle$, the $\gamma$-th element is given by $((-1)^{\delta_1} P_1)^{\gamma_1} \cdots ((-1)^{\delta_n} P_n)^{\gamma_n} = (-1)^{\sum_{i=1}^{n} \gamma_i \delta_i} P_\gamma$, whose coefficient matches $(H_n)_{\gamma, \delta} = (-1)^{\sum_{i=1}^{n} \gamma_i \delta_i}$. Thus, for a check matrix $C$, the Walsh--Hadamard matrix $H_n$ arises. By reordering and flipping the sign of each row in $H_n$ according to the Pauli operator and phase of each $P_\gamma$, $W \in \calW_n$ for the check matrix $C$ can be obtained.
    }

    \black{Next, we present a constructive method to enumerate all the check matrices. Note that this allows us to enumerate $\{W_j\}_{j=1}^{\abs{\calS_n}/2^n}$ simply by computing all $P_\gamma$ for a given check matrix $C$ and multiplying by $H_n$. To enumerate check matrices, we use the following standard form of the check matrix representation:
    \begin{align}
        \left(
        \begin{array}{cc|cc}
            I_{k} & \hat{X} & \hat{Z} & O \\
            O & O & \hat{X}^\top & I_{n-k} \\
        \end{array}
        \right),\label{eq:check_matrix_standard_form}
    \end{align}
    where $\hat{X}\in\bbF_2^{k \times (n-k)}$ is given by reduced row echelon form with rank $k$ and $\hat{Z}=\hat{Z}^\top \in \bbF_2^{k\times k}$ is a symmetric matrix.
    Since all the choices of $\hat{X}$ and $\hat{Z}$ produce mutually distinct check matrices that satisfy both Lemmas~\ref{lemma:check_matrix_independent} and~\ref{lemma:check_matrix_commutativity}, we can ensure that Eq.~\eqref{eq:check_matrix_standard_form} yields a unique construction of the set of check matrices.
    We can verify that this 
 method indeed constructs all the check matrices.
    The number of choices for $\hat{Z}$ is $2^{k(k+1)/2}$, and for $\hat{X}$ it is given by the $q$-binomial coefficient $\qBinom{n}{k}_2$~\cite[Theorem 7.1]{kacQuantumCalculus2002}, defined for a general $q (\neq 1)$ as
    \begin{equation*}
    \qBinom{n}{k}_q = \frac{(1-q^n)(1-q^{n-1})\dots (1-q^{n-k+1})}{(1-q)(1-q^2)\dots(1-q^k)}.
    \end{equation*}
    Now, the $q$-binomial theorem~\cite{struchalinExperimentalEstimationQuantum2021a} shows $\sum_{k=0}^{n} \qBinom{n}{k}_2 2^{k(k+1)/2} = \prod_{k=0}^{n-1} (2^{n-k} + 1) = \abs{\calS_n}/2^n$. This verifies that the constructive method enumerates all the stabilizer states, which \mbox{completes} the proof.
    }
\end{proof}
We utilized the Gray code to enumerate the matrix representations in the actual numerical implementation. Refer to the codes available on GitHub~\cite{yoshiokaQuantumprogrammingRoMhandbookHandbook2023} for details. Now, by applying Lemma~\ref{lem:Walsh--Hadamard} for each $W_j$ in Lemma~\ref{lem:Amat_decomposition}, we complete the proof of Theorem~\ref{thm:inner_product}.

\black{We can compute the maximum fidelity as well.
The fidelity between quantum states $\sigma_j=\ketbra{\phi_j}{\phi_j}$ and $\rho$ is defined as
\begin{equation*}
    F(\sigma_j, \rho) \defeq \Tr\qty[ \sqrt{\sigma_j^{1/2} \rho \sigma_j^{1/2}}]
    = \Tr\qty[\sqrt{|\phi_j \rangle\langle \phi_j | \rho |\phi_j \rangle\langle \phi_j |}]
    = \sqrt{\langle \phi_j | \rho |\phi_j \rangle}.
\end{equation*}
Since $\bm{b}$ is a Pauli vector of $\rho$,  we have $\rho = \frac{1}{2^n} \sum_{i=1}^{4^n} b_i P_i$ and
\begin{equation*}
    F^2(\sigma_j, \rho)
    = \langle \phi_j | \rho |\phi_j \rangle
    = \Tr[\sigma_j \rho]
    = \Tr[\qty(\frac{1}{2^n} \sum_{i=1}^{4^n} (\bm{a}_j)_{i} P_i)\qty(\frac{1}{2^n} \sum_{i=1}^{4^n} b_i P_i)]
    = \frac{1}{2^n} \bm{a}_{j}^\top \bm{b}.
\end{equation*}
In the last equality, we used the orthogonality of Pauli matrices.
Based on this relationship, we can derive the following corollary of Theorem~\ref{thm:inner_product}.
\begin{corollary}[Computing maximum fidelity]
Maximum fidelity for $\rho$ can be computed with the time complexity of $\order{n\abs{\calS_n}}$ and the space complexity of $\order{2^n}$ excluding inputs.
\end{corollary}
For the case when $\rho=\ketbra{\psi}{\psi}$ is a pure state, we can compute the stabilizer fidelity~\cite{bravyiSimulationQuantumCircuits2019} as
\begin{equation*}
    F_{\rm STAB}(\psi) \defeq \max_{\ket{\phi_j} \in \calS_n} \abs{\braket{\phi_j}{\psi}}^2 = \frac{1}{2^n} \qty(\max_{1 \leq j \leq \abs{\calS_n}} \bm{a}_j^\top \bm{b}).
\end{equation*}}
It has been recognized that the stabilizer fidelity cannot be computed with a moderate computational cost for $n>5$~\cite{haugEfficientStabilizerEntropies2023}. Meanwhile, our algorithm allowed us to compute up to $n=8$ in 4 hours.
This numerical experiment was conducted using \Cpp17 compiled by GCC 9.4.0 and a \textit{cluster computer} powered by Intel(R) Xeon(R) CPU E5-2640 v4 with \SI{270}{\giga\byte} of RAM using 40 threads.
Even if we use a \textit{laptop} powered by Intel(R) Core(TM) i7-10510U CPU with \SI{16}{\giga\byte} RAM, we can compute $n=7$ in 2 minutes using 8 threads.

\subsection{Top-Overlap Method for RoM}\label{subsec:top_K}

Using the overlap calculation subroutine presented in Sec.~\ref{subsec:overlap}, we propose a novel algorithm that computes the approximate or exact RoM by utilizing the following properties:
(i) by nature of $L^1$ norm minimization problem, the solution of the optimal stabilizer decomposition is sparse~\cite{strangLinearAlgebraLearning2019b}, and (ii) the stabilizer overlap is closely related to the optimal stabilizer decomposition (see Fig.~\ref{fig:dot_and_coeff}(a)).
Concretely, as we provide the detail in Algorithm~\ref{alg:main}, we restrict the number of columns in $\Amat_n$ and consider only a fraction $K$ of stabilizer states with the largest or smallest overlaps; the fraction of $1-K$ is neglected.
\begin{algorithm}[ht]
    \KwIn{Pauli vector $\bm{b}$ for the quantum state $\rho$, fraction $K~(0 < K \leq 1)$}
    \KwOut{An approximate or exact RoM}

    \SetKwFunction{SolveLP}{SolveLP}

    Compute the stabilizer overlaps $\Amat_n^{\top} \bm{b}$ using the FWHT algorithm.\\
    $\calC_0 \gets$ Partial set of $\calA_n$ with size of $K\abs{\calS_n}$ using the largest and smallest overlaps.\\
    $\bm{x}_0, \bm{y}_0 \gets \SolveLP(\calC_0, \bm{b})$\\
    \Return $\hat{\calR}_0(\rho) = \norm{\bm{x}_0}_1$
    
    \caption{Top-Overlap Method for RoM}
    \label{alg:main}
\end{algorithm}
\begin{figure}[t]
    \begin{center}
        \includegraphics[width=0.95\columnwidth]{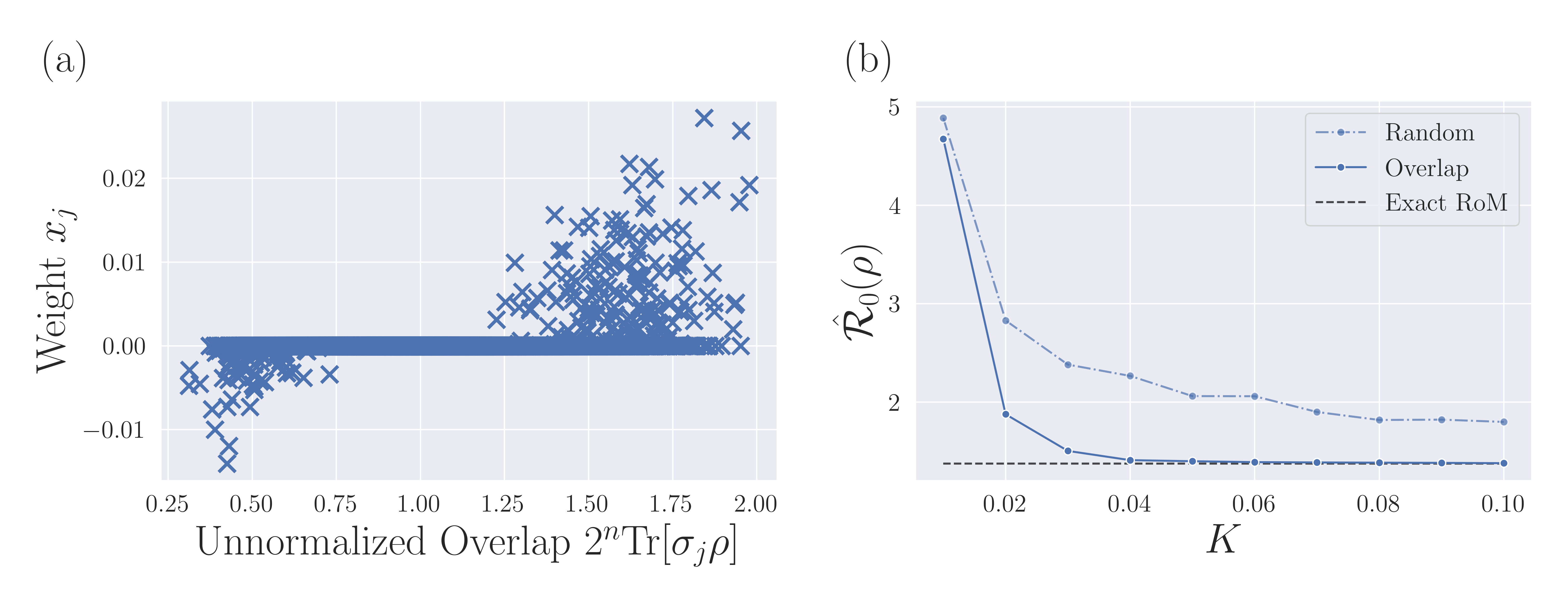}
        \caption{
        (a) Stabilizer overlaps $2^n\Tr[\sigma_j \rho]$ and the weights in the primal solution $x_j$ for a Haar random mixed 4-qubit state $\rho$.
        (b) The computed value $\hat{\calR}_0(\rho)$ in Algorithm~\ref{alg:main} with a restricted set of stabilizers $\calC_0$ of fraction $0 < K \leq 1$ for the same state $\rho$.
        }
        \label{fig:dot_and_coeff}
    \end{center}
\end{figure}

We can confirm the observation (ii) in Fig.~\ref{fig:dot_and_coeff}(a). We show the distribution of stabilizer overlaps $\bm{a}_j^\top \bm{b}=2^n\Tr[\sigma_j \rho]$ and their weights in the primal solution $x_j$ for Haar random 4-qubit mixed state $\rho$. Indeed, we find strong correspondence between the stabilizer overlaps and weights in various instances.
\black{An intuitive explanation for this property can be provided by recalling the definition of RoM in Eq.~\eqref{eq:RoMDefOrig}, in which two states in $\mathrm{STAB}_n$ are used: $\rho_+$ and $\rho_-$. As illustrated in the visualization of RoM~\cite[Appendix A]{howardApplicationResourceTheory2017a}, $\rho_+$ resembles $\rho$ as much as possible, while $\rho_-$ diverges from $\rho$ as much as possible. States with large overlaps support $\rho_+$, whereas states with small overlaps support $\rho_-$, indicating this tendency.}

As we highlight in Fig.~\ref{fig:dot_and_coeff}(b), we only need a small fraction $\calC_0$ from the entire $\calA_n$ to obtain a nearly exact solution $\hat{\calR}_0(\rho) \approx \calR(\rho)$ in Algorithm~\ref{alg:main}.
We find that, for a Haar random mixed state of $n=4$ qubit system, it is sufficient to use a fraction of $K \defeq \abs{\calC_0}/\abs{\calA_n} = 0.05$ to achieve an absolute error of 0.023. The value of $K$ required to achieve similar accuracy for larger systems are $K=10^{-2}, 10^{-3}, 10^{-5}$ for $n=5,6,7$, respectively.
Meanwhile, we must take a significantly larger column set to obtain the exact RoM; the fraction should be $K\sim 0.32$ for $n=4$ qubit case.
We provide further numerical details in Appendix~\ref{app:numerical_details}. 

\subsection{Column Generation Method for Dualized RoM} \label{subsec:column_generation}

Despite the significant improvement over the naive method, Algorithm~\ref{alg:main} contains some issues: (i) we cannot tell whether the column set $\calC_0 \subset \calA_n$ is sufficient to yield the exact RoM,
(ii) there is no quantitative measure to judge the quality of the solution, and (iii) there is a large gap between ``highly approximate'' and ``exact'' RoMs in the computational resource.
These issues are well addressed when we consider dualized formalism instead. Recall that the dualized formulation of RoM in Eq.~\eqref{eq:RoM_dual} is given by
\begin{equation*}
    \calR(\rho) = \max_{\bm{y}\in \bbR^{4^n}} \left\{ \bm{b}^{\top} \bm{y} \relmiddle| -\bm{1} \leq \Amat_n^\top \bm{y} \leq \bm{1} \right\}.
\end{equation*}
If there exists $\bm{a} \in \calA_n$ such that violates the constraint in Eq.~\eqref{eq:RoM_dual}, $|\bm{a}^{\top} \bm{y}| > 1$, then we must use those violated $\bm{a}$'s to attain the exact RoM (Also refer to Fig.~\ref{fig:imageOfDual}).
Conversely, without any violation, the solution is exact owing to the problem's strong duality.
This motivates us to derive Algorithm~\ref{alg:column_generation} using the Column Generation (CG) technique~\cite{desaulniersColumnGeneration2005}.

\begin{figure}[t]
    \centering
    \includegraphics[width=\linewidth]{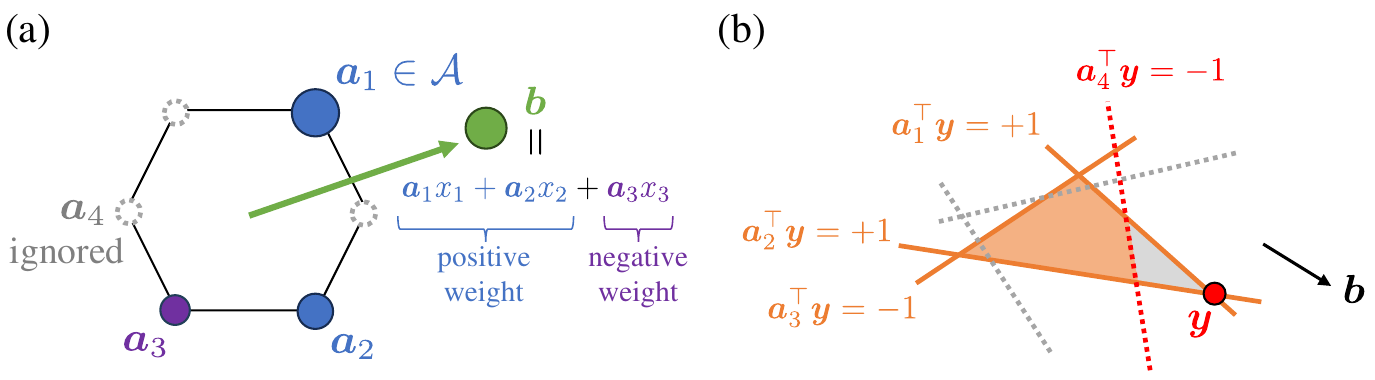}
    \caption{
    (a) Graphical description of the primal formalism. The Pauli vector $\bm{b}$ of the target state is decomposed into a sum over those of pure stabilizer states $\bm{a}$ denoted by vertices on the stabilizer polytope to minimize $\|\bm{x}\|_1$.
    The restriction on columns of $\Amat_n$ is expressed by gray vertices, which are eliminated from the decomposition.
    (b) Graphical description of the dual formalism. The equality constraints in the primal problem are now given as inequalities $-\bm{1} \leq \Amat_n^\top \bm{y} \leq \bm{1}$ for the dual variable $\bm{y}$, while the objective function is the inner product with $\bm{b}$. A solution $\bm{y}$ obtained from reduced column set $\calC$ (orange) may violate some constraints (red). So, one should add those columns to improve the solution, while some columns denoted by gray dotted lines do not affect the result.}
    \label{fig:imageOfDual}
\end{figure}
\begin{algorithm}[ht]
    \KwIn{Pauli vector $\bm{b}$ of the target state $\rho$}
    \KwOut{Exact RoM $\calR(\rho)$}

    \SetKwFunction{SolveLP}{SolveLP}
    
    $\calC_0 \gets$ Partial set of $\calA_n$ with size of $K\abs{\calS_n}$ using the largest and smallest overlaps.\\
      \For{$k = 0, 1, 2, \ldots$} {
        $\bm{x}_k, \bm{y}_k \gets \SolveLP(\calC_k,\bm{b})$\\
        $\hat{\calR}_k(\rho) \gets \norm{\bm{x}_k}_1$\\
        $\calC' \gets \left\{ \bm{a}_j \in \calA_n \relmiddle| \abs{\bm{a}_j^{\top} \bm{y}_k} > 1 \right\}$
        \Comment*[r]{Use of FWHT}
        \If {$\calC' = \emptyset$} {
            \Return $\calR(\rho)=\hat{\calR}_k(\rho)$
        }
        $\calC_{k+1} \gets \calC_k \cup \calC'$ \Comment*[r]{Two modifications in the main text}
    }
    \caption{Exact RoM Calculation by Column Generation}
    \label{alg:column_generation}
\end{algorithm}

\begin{table}[t]
\begin{center}
\begin{tabular}{c|ccc}
\hline
Qubit count $n$ & Full-size $\Amat_n$ & Our work & $K$ \\ \hline
        4 & \SI{3}{\mebi\byte} & \SI{301}{\kibi\byte} & $10^{-1}$ \\
        5 & \SI{379}{\mebi\byte} & \SI{4}{\mebi\byte} & $10^{-2}$ \\
        6 & \SI{95}{\gibi\byte} & \SI{97}{\mebi\byte} & $10^{-3}$ \\
        7 & \SI{86}{\tebi\byte} & \SI{499}{\mebi\byte} & $10^{-5}$ \\
        8 & \SI{86}{\pebi\byte} & \SI{512}{\mebi\byte} & $10^{-8}$ \\  \hline
\end{tabular}
    \caption{Memory size necessary for storing the entire $\calA_n$, i.e., the matrix $\bm{A}_n$, and the reduced matrix for $\calC_0$.
    $\calC_0$ extracts only $K\abs{\calA_n}$ columns from $\calA_n$ and is used in Algorithm~\ref{alg:main} or the initialization step in Algorithm~\ref{alg:column_generation}. This data size information relies on the sparse matrix format in SciPy~\cite{scipyScipySparseCsc_matrix}.}
    \label{table:sizeOfAn}
\end{center}
\end{table}

The CG technique means to compute the overlap $\abs{\bm{a}_j^{\top} \bm{y}_k}$ for all $\bm{a}_j \in \calA_n$ at each iteration to find and generate the columns violating the constraint. By iteratively updating $\calC_k$ until there is no violation, we obtain the exact RoM. As in Table~\ref{table:sizeOfAn}, We initialized the partial column set $\calC_0$ with $K=10^{-5}, 10^{-8}$ for $n=7, 8$, respectively.
In Fig.~\ref{fig:CG_7_8}, we present the numerical experiment results of Algorithm~\ref{alg:column_generation}.
The number of violated columns decreased rapidly, and the exact RoM was obtained after a small number of iterations. The run time was 2 hours using the laptop for $n=7$ and 2 days using the cluster computer for $n=8$.
Although the algorithm is not guaranteed to yield exact solutions within a realistic run time, we expect the exact RoM to be obtained within a similar run time for almost all cases. See Appendix~\ref{app:numerical_details} for more detailed results.

\begin{figure}[t]
    \begin{center}
        \begin{minipage}{0.49\hsize}
            \includegraphics[width=\columnwidth]{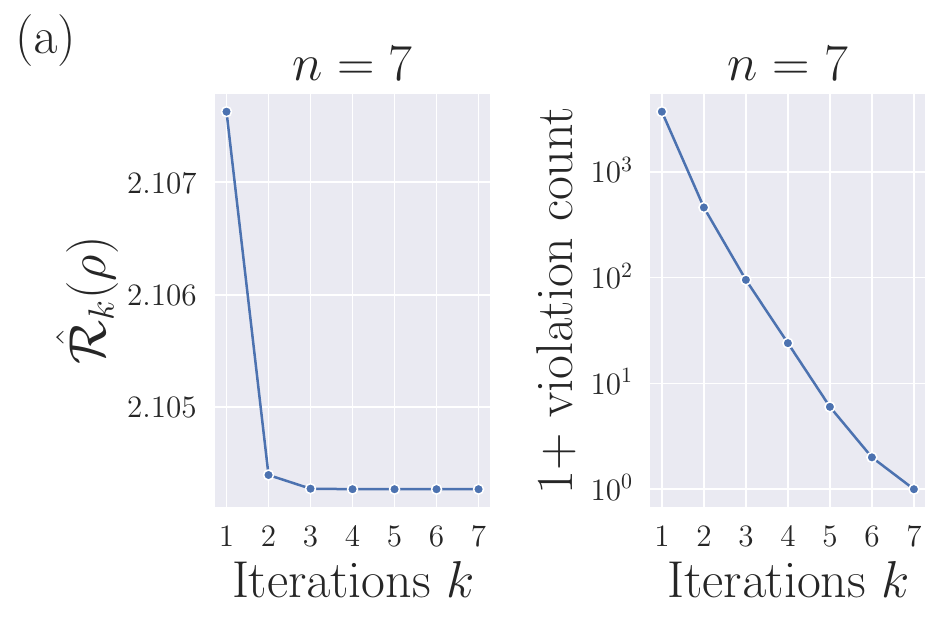}
        \end{minipage}
        \begin{minipage}{0.49\hsize}
            \includegraphics[width=\columnwidth]{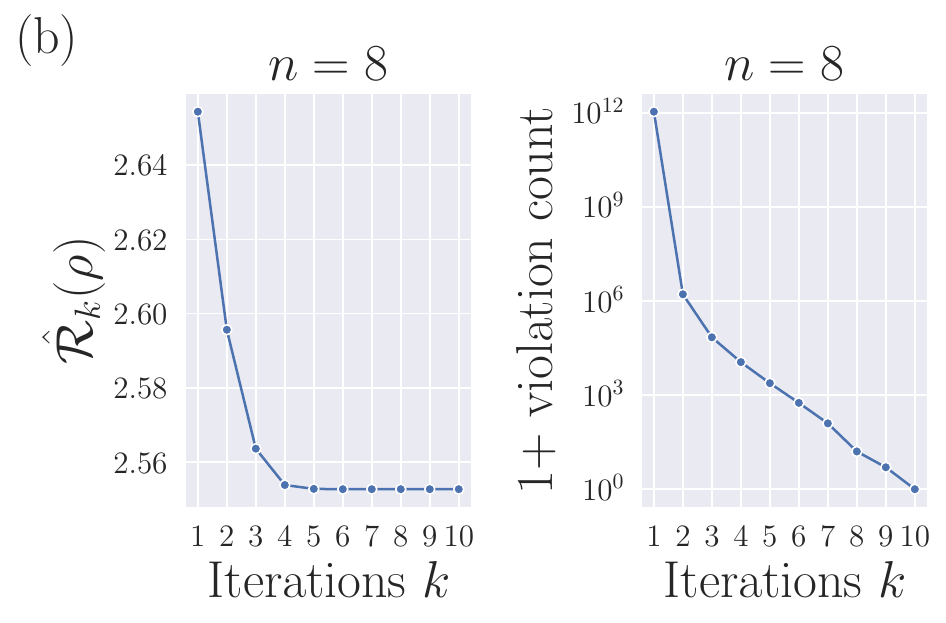}
        \end{minipage}
        \caption{
            Demonstration of Algorithm~\ref{alg:column_generation} for Haar random mixed states of (a) $n=7$ qubits and (b) $n=8$ qubits.
            We display the value of $\hat{\calR}_k(\rho)$ in Algorithm~\ref{alg:column_generation} and the number of violated inequality conditions ($\abs{\calC'}$).
            $\abs{\calC'}$ quickly decreases to zero, which assures that the exact RoM is obtained.
            }
        \label{fig:CG_7_8}
    \end{center}
\end{figure}

We remark two practical modifications in line 8 of Algorithm~\ref{alg:column_generation}.
First, we introduce a threshold $d \: (0<d<1)$, such as 0.8, to discard columns $\bm{a}_j \in \calC_k$ which satisfy both $\abs{\bm{a}_j^\top \bm{y}_k} < d$ and $(\bm{x}_k)_j = 0$. Such columns do not affect the solution with high probability.
Second, instead of adding the entire violating columns $\calC'$ to $\calC_{k+1}$, we set an upper bound on the number of the columns to add in order to suppress the memory consumption. For the calculation in Fig.~\ref{fig:CG_7_8}, we have added only $K|\mathcal{S}_n|$ columns with either the largest or smallest overlaps at each iteration, significantly suppressing the memory consumption.

\subsection{Minimal Feasible Solution With Accuracy Guarantee}\label{subsec:minimal}

Exact solutions for $n \geq 9$ qubit systems may require prohibitively enormous computational resources, while we may still wish to compute a feasible solution. In this section, we propose a method with minimal computational resources that always guarantees to yield a feasible solution. 
Consequently, we can obtain an approximate RoM with its corresponding feasible solution ${\bm x}$ for any state of $n=14$ qubits within a minute.

\black{The method's main idea is to construct a set $\calV_n=\{W_1, \dots, W_{2^n + 1}\} \subset \calW_n$ that always yields a feasible solution. 
We define a matrix $\bm{M}_n$ as
\begin{equation*}
    \bm{M}_n \defeq \mqty[W_1&\cdots&W_{2^n+1}],
\end{equation*}
and for $W_j \in \calV_n$, we denote the set of all the nonzero rows of the matrix $W_j$ by $R(W_j)$.
Then, $\texttt{Prob}(\bm{M}_n,\bm{b})$ guarantees a feasible solution if and only if $\qty{R(W_j) \mid W_j \in \calV_n}$ is a cover of all the rows, i.e., $\bigcup_{W_j \in \calV_n} R(W_j)$ equals the set of all the rows. The necessity is trivial, and the sufficiency is later discussed in Theorem~\ref{thm:cover_matrix}. We refer to $\calV_n$ as a cover set and $\bm{M}_n$ as a cover matrix.
The following proposition assures the existence of the cover set; we prove it in Appendix~\ref{app:proof_proposition_cover_set} by taking a set of mutually unbiased bases~\cite{Wootters1989Optimal, GibbonsHW2004Discrete} as $\calV_n$.
\begin{restatable}{proposition}{coverset}\label{prop:cover_set}
    A cover set $\calV_n$ with size $2^n + 1$ can be constructed.
\end{restatable}
As a remark, the size $2^n+1$ is minimum to ensure the feasibility.
The nonzero-row set $R(W)$ covers $2^n$ rows, and the number of rows to be covered is $4^n$. Since all the nonzero-row sets $R(W)$ shares the first row corresponding to $I^{\otimes n}$, the size of the cover set $\calV_n$ should be equal or larger than $(4^n - 1) / (2^n - 1) = 2^n + 1$.}

\begin{figure}[t]
    \begin{center}
        \includegraphics[width=\columnwidth]{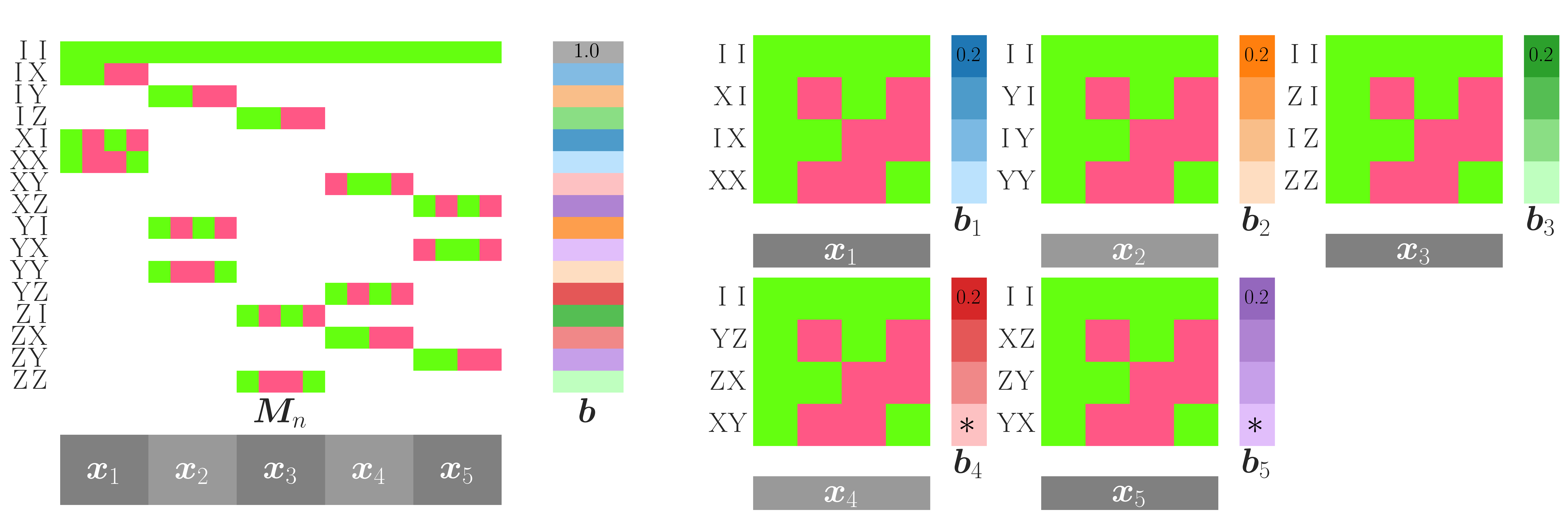}
        \caption{
        Graphical description of the cover matrix $\bm{M}_n$ and $\bm{b}_j$.
        The first element of each $\bm{b}_j$ is $1/(2^{n}+1)$, and the star mark (*) means the sign flip.
        We can execute the FWHT algorithm to solve $H_n \bm{x}_j = \bm{b}_j$.
        }
        \label{fig:illustrationOfSolutionByFWHT}
    \end{center}
\end{figure}
Fig.~\ref{fig:illustrationOfSolutionByFWHT} shows an example of the cover matrix with $n = 2$. Since the cover matrix uses the minimum possible number of blocks, only one block covers each row except for the first one. The following theorem exploits the structure of the cover matrix to obtain a feasible solution.
\black{
\begin{theorem}[Minimal Feasible Solution]\label{thm:cover_matrix}
    A feasible solution for $\texttt{Prob}(\bm{M}_n, \bm{b})$ can be obtained with time complexity of $\order{n4^n}$.
\end{theorem}
\begin{proof}
   Our goal is to obtain a feasible solution for $\texttt{Prob}(\bm{M}_n, \bm{b})$, in other words, $\bm{x}$ such that $\bm{M}_n \bm{x} = \bm{b}$.
    This can be obtained by solving the simultaneous equations $W_j \bm{x}_j = \bm{b}_j' \: (1 \leq j \leq 2^n+1)$ where $\bm{x}_j \in \bbR^{2^n}$ and $\bm{b}_j' \in \bbR^{4^n}$ which is defined as
    \begin{equation*}
    \text{The $i$-th element of $\bm{b}_j'$} \defeq \begin{cases}
    1/(2^n+1) &\text{if $i=0$,}\\
    \text{the $i$-th element of $\bm{b}$} &\text{if $i$-th row in $R(W_j)$,}\\
    0 &\text{otherwise}.
    \end{cases}
    \end{equation*}
    Thus, by defining $\bm{b}_j \in \bbR^{2^n}$ as the nonzero part of $\bm{b}_j'$ with reordering and sign flip so that $W_n$ coincides with $H_n$, we have $H_n \bm{x}_j = \bm{b}_j \iff \bm{x}_j = \frac{1}{2^n}H_n\bm{b}_j$. By combining each $\bm{x}_j$ we can construct $\bm{x}$. The time complexity is $\order{(2^n+1) n 2^n}$ by the FWHT, i.e., $\order{n4^n}$.
\end{proof}
}

\black{The solution obtained in Theorem~\ref{thm:cover_matrix} is evidently feasible not only for $\texttt{Prob}(\bm{M}_n,\bm{b})$ but also for $\texttt{Prob}(\bm{A}_n, \bm{b})$, which provide an approximate value for the RoM, $R_\mathrm{FWHT} \defeq \norm{\bm{x}}_1$.
By utilizing the st-norm $\stnorm{\rho}=\frac{1}{2^n}\norm{\bm{b}}_1$ mentioned in Sec.~\ref{subsec:dualizedRoM}, we can derive
\begin{equation*}
    R_\mathrm{FWHT} = \sum_{j=1}^{2^n+1}\norm{\bm{x}_j}_1 = \sum_{j=1}^{2^n+1}\frac{1}{2^n}\norm{H_n \bm{b}_j}_1 \leq \sum_{j=1}^{2^n+1} \norm{\bm{b}_j}_1 = 2^n \qty(\frac{1}{2^n} \norm{\bm{b}}_1) = 2^n\stnorm{\rho} \leq 2^n\calR(\rho).
\end{equation*}
Now, let us assume that each element $b_{j,i}$ of $\bm{b}_j$ is normally distributed with a mean of 0, which roughly approximates the distribution of a Haar random state.
Then, since $\sum_{i=1}^{2^n} b_{ji}$ follows the distribution with a $2^{n/2}$ times standard deviation of $b_{ji}$, we have
\begin{equation*}
    \bbE \left[ \frac{1}{2^n} \norm{H_n \bm{b}_j}_1 \right]
    = \bbE \left[ \abs{\sum_{i=1}^{2^n} b_{j,i}} \right] 
    = 2^{n/2} \bbE \left[ \abs{b_{j,i}} \right] = \frac{1}{2^{n/2}} \bbE \left[ \norm{\bm{b}_j}_1 \right],
\end{equation*}
which allows us to anticipate that $R_\mathrm{FWHT} \approx 2^{n/2}\stnorm{\rho}$.}

\black{Empirically, we confirmed that 100 Haar random 10-qubit mixed states satisfy
\begin{equation*}
    0.994 \leq \frac{R_\mathrm{FWHT}}{2^{n/2}\stnorm{\rho}} \leq 1.002.
\end{equation*}
This indicates $R_{\mathrm{FWHT}} \leq 2^{n/2}\calR(\rho)+\order{1}$ with high probability. While $R_\mathrm{FWHT}$ deviates from the exact RoM, it is useful in certain scenarios, such as providing an initial solution for exact RoM calculation or assuring the feasibility of the algorithms (see Appendix~\ref{subsec:AssureFeasibility}).}

\section{Quantum Resource of Multiple Magic States}\label{sec:application}
One of the important applications of RoM is measuring the total nonstabilizerness of multiple magic states decoupled from each other. 
For instance, one may wish to evaluate the amount of magic resources for copies of states $\rho^{\otimes n}$ to estimate the upper bound on the number of generatable clean magic states.
This could even include situations where quantum states are nonequivalent, such as partially decoupled states $\bigotimes_i \rho_i$.
There exists a previous work by Heinrich and Gross~\cite{heinrichRobustnessMagicSymmetries2019} that has utilized the symmetry of some pure magic states such as $\ketbra{H}{H} = \frac{1}{2}\left( I + \frac{1}{\sqrt{2}}(X+Y)\right)$ and $\ketbra{F}{F} = \frac{1}{2}\left( I + \frac{1}{\sqrt{3}}(X+Y+Z)\right)$ to scale up the simulation up to $n=26$ qubits, while we still lack a method to investigate general quantum states with partial disentangled structure (see Fig.~\ref{fig:application}).

In this section, we apply the algorithms proposed in Sec.~\ref{sec:scale_up} to practical problems: copies of identical quantum states $\rho^{\otimes n}$ and partially disentangled quantum states $\bigotimes_i \rho_i.$
In particular, we first discuss the case of permutation symmetric state $\rho^{\otimes n}$ in Sec.~\ref{subsec:perm_symm}, and then also consider general partially disentangled states in Sec.~\ref{subsec:general_partial_disentangled}.

\begin{figure}[tb]
    \begin{center}
        \includegraphics[width=0.75\columnwidth]{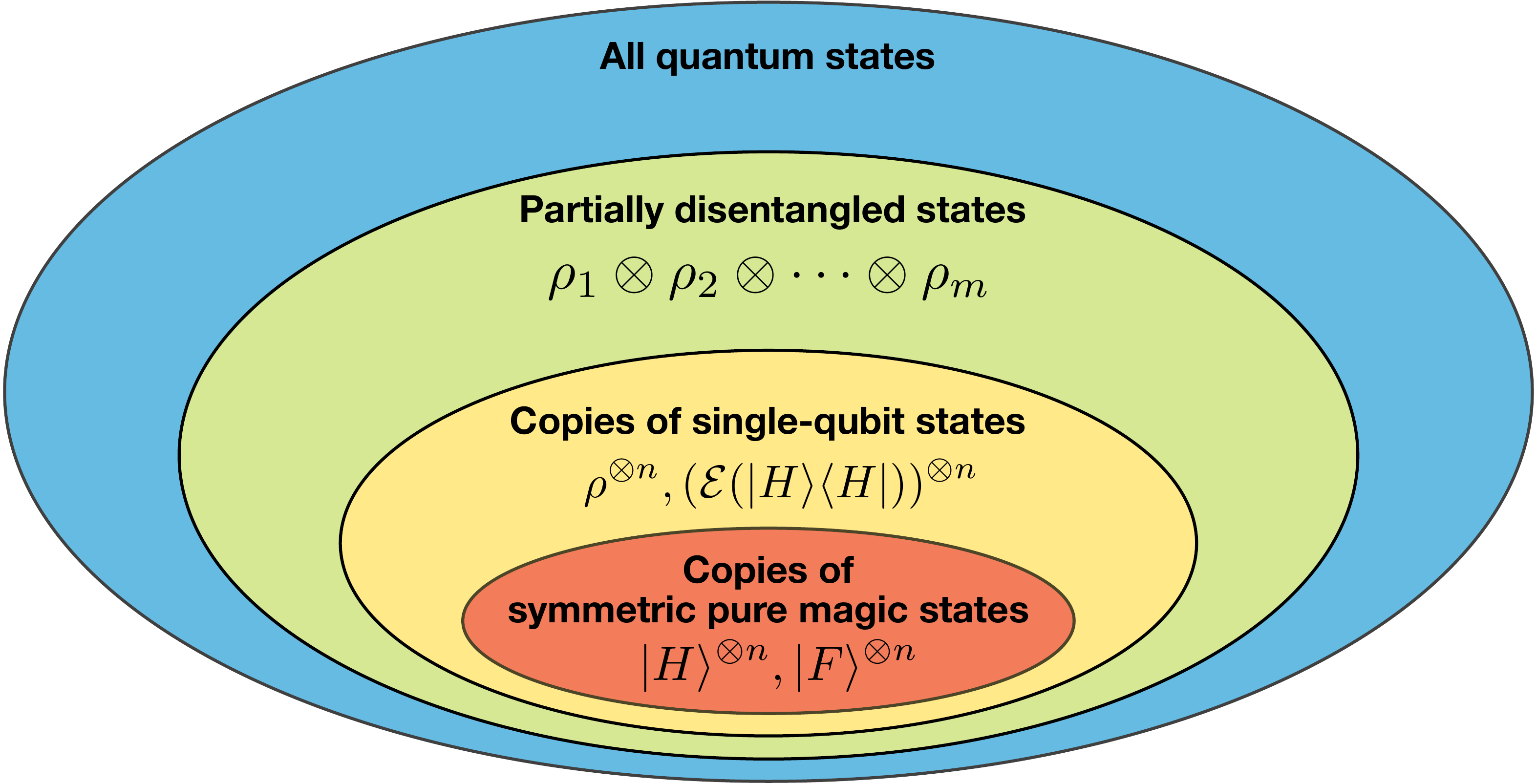}
        \caption{
        Hierarchy of partially disentangled quantum states.
        }
        \label{fig:application}
    \end{center}
\end{figure}

\begin{figure}[tb]
    \centering
    \includegraphics[width=0.8\linewidth]{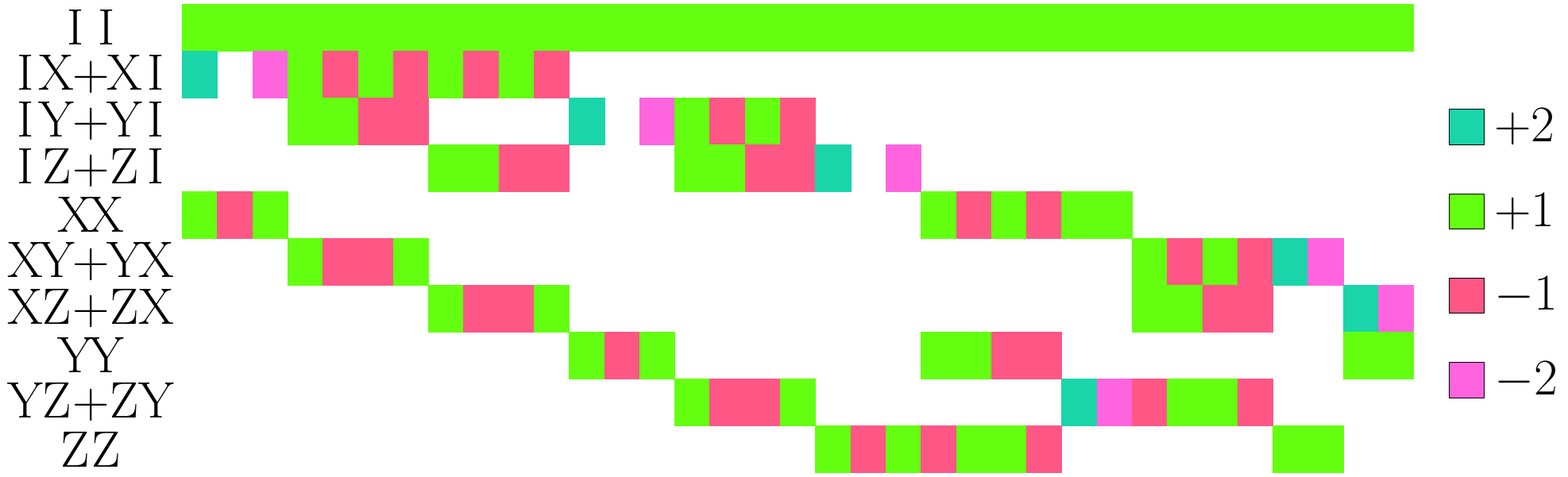}
    \caption{Visualization of $\bm{Q}_n$ for $n=2$.
    Compare with Fig.~\ref{fig:Amat_2} to confirm how the matrix is compressed.}
    \label{fig:perm}
\end{figure}

\subsection{Copies of Single-Qubit States}\label{subsec:perm_symm}
When the target quantum state is given as identical copies of a quantum state $\rho^{\otimes n}$, we may compress the size of $\Amat_n$ by utilizing the permutation symmetry to combine multiple rows and columns of $\calA_n$.
In this work, we employed the compression method for $\mathcal{A}_n$ proposed in Ref.~\cite{heinrichRobustnessMagicSymmetries2019} to define a set of permutation symmetric columns $\calQ_n$  (see Fig.~\ref{fig:perm}).
As in Ref.~\cite{heinrichRobustnessMagicSymmetries2019}, we also used the data by Danielsen~\cite{larseirikDatabaseSelfDualQuantum}, which enables us to obtain the exact solutions for $n\leq 7$ qubits.
Beyond $n=7$ or 8 qubits, although a matrix $\bm{Q}_n$ of $\calQ_n$ and $\bm{b}^{\otimes n}$ are reduced by permutation symmetry, it is not realistic to obtain the exact solution.

Therefore, we propose an approximate method with the divide-and-conquer method.
In the following, we use \texttt{Prob} and \texttt{SolveLP} for the compressed objects with the same meaning. \black{Also, the tensor product of the compressed vectors is defined as the compressed vector of the tensor product of the uncompressed operands.}
In Algorithm~\ref{alg:symmetric_reduction_approx}, we consider all possible decomposition of $m$ qubits into two groups with $j$ and $k$ qubits~($j+k=m$). Since we store the solutions for $\texttt{Prob}(\bm{Q}_i, \bm{b}^{\otimes i})$ for $i < m$, we can load the results.
We take the union of tensor product of columns with nonzero weights as $\bigcup_{j+k=m} \{\bm{q}_j \otimes \bm{q}_k \mid \bm{q}_j \in \calC_j, \bm{q}_k \in \calC_k\}$ to construct $\calC_m$, then we compute the approximate RoM with $\calC_m$, which should be less than the product of RoM computed for each subsystem.
\begin{algorithm}[ht]
    \caption{Approximate RoM for Permutation Symmetric States}
    \label{alg:symmetric_reduction_approx}
    
    \KwData{Compressed column set $\calQ_n$ of matrix $\bm{Q}_n$}
    \KwIn{Positive integer $n, k$ ($n \geq k$), Pauli vector $\bm{b}$ for the target state $\rho$}
    \KwOut{Approximate RoM $R_i$ of $\rho^{\otimes i}$  ($i = 1, \dots, n$)}
    \SetKwFunction{SolveLP}{SolveLP}    
    \For{$i \gets 1$ \KwTo $k$}{
        $\bm{x}_i, \bm{y}_i \gets \SolveLP(\calQ_i, \bm{b}^{\otimes i})$\\
        $R_i \gets \norm{\bm{x}_i}_1$\\
        $\calC_i \gets$ columns $\{\bm{q}_j\} \subseteq \calQ_i$ where $(\bm{x}_i)_j$ is nonzero.\\
    }
    \For{$i \gets k + 1$ \KwTo $n$}{
        $\calC' \gets \emptyset$\\
        \For{$l \gets 1$ \KwTo $\lfloor i / 2 \rfloor$}{
            $m \gets i - l$\\
            $\calC' \gets \calC' \cup \left\{\bm{q}_l \otimes \bm{q}_m \relmiddle| \bm{q}_l \in\calC_l, \bm{q}_m \in \calC_m\right\}$\\
        }
        $\bm{x}_i, \bm{y}_i \gets \SolveLP(\calC', \bm{b}^{\otimes i})$\\
                $R_i \gets \norm{\bm{x}_i}_1$\\
        $\calC_i \gets$ columns $\{\bm{q}_j\} \subseteq \calC'$ where $(\bm{x}_i)_j$ is nonzero.\\
    }
    \Return $(R_1, \dots, R_n)$\;
\end{algorithm}

Figure~\ref{fig:mixed_random_symmetry} shows the results of a numerical demonstration of the proposed algorithm applied to Haar random pure state, mixed state, and copies of pure magic state $\ket{H}$.
Using the exact stabilizer decomposition up to $k=7$ qubits, we have successfully computed the approximate RoM up to $n=17$ for the pure and mixed random state. For $\ket{H}$, the compressed column set size $\abs{\calC_n}$ is significantly smaller so that we have reached $n=21$. While this is not as large as $n=26$ reported in Ref.~\cite{heinrichRobustnessMagicSymmetries2019}, we emphasize that the present work is based on an algorithm that is agnostic to the internal symmetry of the single-qubit state.
As a remark, the approximate RoM for copies of pure magic states $\ket{H}$ are almost identical to those presented in Ref.~\cite{heinrichRobustnessMagicSymmetries2019}; the value was at most 1.007 times larger.

\begin{figure}[t]
    \begin{center}
        \includegraphics[width=0.85\columnwidth]{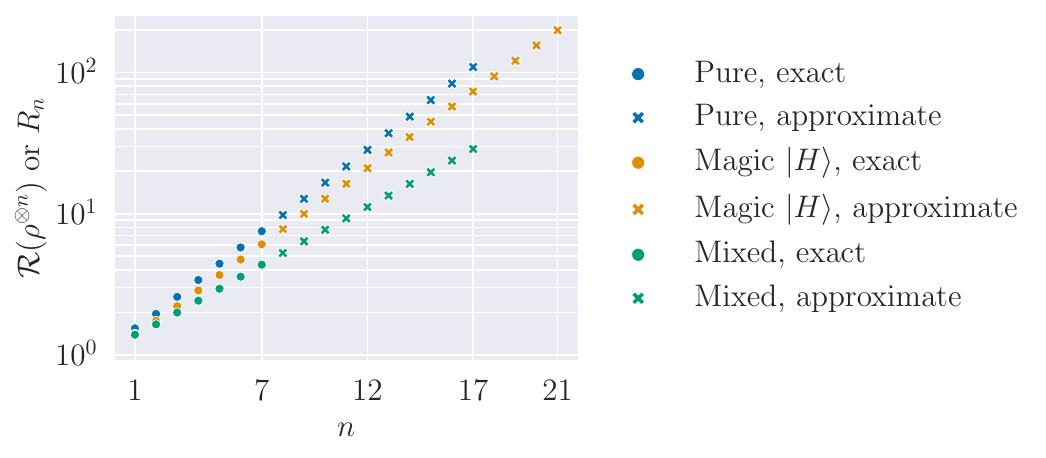}
        \caption{
        Exact $\calR(\rho^{\otimes n})$ or approximate $R_n$ computed in Algorithm~\ref{alg:symmetric_reduction_approx}.
        The blue, orange, and green data correspond to the Haar random pure state, the magic state $\ket{H}$, and the Haar random mixed state, respectively.
        The circle or crossed points indicate whether the solution is exact or approximate. The approximate values are computed from the exact solutions for $n\leq 7$ qubits.
        }
        \label{fig:mixed_random_symmetry}
    \end{center}
\end{figure}

\subsection{Partially Disentangled States} \label{subsec:general_partial_disentangled}

Let us assume that we are interested in a general partially decoupled state of $m$ subsystems as $\rho = \bigotimes_{i=1}^m \rho_i$, where $\sum_{i=1}^m n_i = n$ with $n_i$ being the qubit count of $i$-th subsystem.
Similar to the previous section, we may first compute the optimal solutions for each subsystem and then take tensor products over nonzero-weight stabilizers to construct the reduced basis for the total system.
Although the computed value is not necessarily the exact RoM, this is an optimal solution in the following meaning. Refer to Appendix~\ref{app:proof_proposition_multi} for the proof.
\begin{restatable}{proposition}{OptimalityForTheRestrictedProblem}\label{prop:multiplicativity}
    Let $\rho_{i}$ be given for the $i$-th subsystem, 
    $\calR(\rho_i)$ and $\bm{x}_i$ be the optimal solution and primal variable for $\texttt{Prob}(\Amat_{n_i}, \bm{b}_i)$.
    Then,
    \begin{equation*}
        R = \prod_{i=1}^m \calR(\rho_i), \ \ \bm{x} = \bigotimes_{i=1}^m \bm{x}_i,
    \end{equation*}
    is one of the optimal solutions for $\texttt{Prob}(\bigotimes_{i=1}^m \Amat_{n_i}, \bm{b})$.
\end{restatable}

By noting the submultiplicativity of RoM~\cite{howardApplicationResourceTheory2017a}, i.e., $\calR(\bigotimes_i \rho_i) \leq \prod_{i=1}^m \calR(\rho_i) = R$, it is natural to expect that one can further improve the approximate RoM by combining some of the systems and solving it directly. For instance, one may group several subsystems so that each group consists of 6 or 7 qubits, compute the RoM for these systems using the CG method, and multiply them to obtain a better approximation of the RoM.

We find that it is more effective to consider multiple variations to divide subsystems. As shown in Algorithm~\ref{alg:RoM_divide}, we divide subsystems into various kinds of groups. By comparing the product of those values, we can take the minimal value as the approximate RoM.
For instance, we consider the case of a 15-qubit system that is decoupled into 5 subsystems as $\rho=\rho_1\otimes \rho_2\otimes \rho_3\otimes \rho_4\otimes \rho_5$, where each $\rho_i$ is a 3-qubit state. One may compute the approximate value as $\calR(\rho_1 \otimes \rho_2) \times \calR(\rho_3 \otimes \rho_4) \times \calR(\rho_5)$ or $\calR(\rho_1) \times \calR(\rho_2\otimes \rho_3) \times \calR(\rho_4 \otimes \rho_5)$ and take the minimal value as the approximate output.
One could also speed up the computation by brute-force parallelization.

\begin{algorithm}[ht]
    \KwIn{Target state $\rho = \bigotimes_{i=1}^{m}\rho_i$}
    \KwOut{Approximate RoM}
    $R \gets \prod_{i=1}^{m} \calR(\rho_i)$\\
    \ForEach{\rm{Partition of a set $\{i\}_{i=1}^{n}$ to $\bigcup_{j} I_j$}}{
        $R \gets \min\qty(R, \prod_{j} {\calR(\bigotimes_{i \in I_j} \rho_i)})$
    }
    \Return $R$\;
    \caption{Optimization via Subsystem Division}
    \label{alg:RoM_divide}
\end{algorithm}

\section{Discussion}\label{sec:discussion}
In this work, we have introduced a systematic procedure for computing RoM. This procedure has proven highly effective, surpassing the state-of-the-art results for various quantum states, including random arbitrary states, multiple copies of single-qubit magic states, and partially disentangled quantum states.
We have presented the core subroutine capable of computing the overlaps, algorithms to compute exact RoM for arbitrary quantum states, and algorithms to incorporate the nature of the target quantum state, such as the permutation symmetry between multiple copies of a state and the partially decoupled structure for inhomogeneous magic resources.

Numerous future directions can be envisioned.
First, it is intriguing to seek generalization to other quantum resource measures.
\black{While the optimization strategy, namely the column generation technique, can be widely applied to $L^p$ optimization problems, the overlap calculation in this work has heavily exploited the property of qubit stabilizer states. 
This means that, our technique may be applied to the case of robustness-like measures for multiqubit systems (such as channel robustness~\cite{seddonQuantifyingMagicMultiqubit2019} and dyadic negativity~\cite{seddonQuantifyingQuantumSpeedups2021}), whereas it is nontrivial whether there exists a fast overlap computation method for  multiqudit case. Considering that generalized Pauli operators for qudits constitute a complete basis for operator space, we expect that algorithm similar to the Fast Walsh-Hadamard Transform (FWHT) is applicable, thus speeding up the evaluation of negativity in multiqudit systems~\cite{pashayanEstimatingOutcomeProbabilities2015}. We note that a recent work on stabilizer extent has shown a technique that does not rely on FWHT~\cite{hamaguchiFasterComputation2024}, which could be a promising direction for rank-based monotones.}
Furthermore, it is nontrivial if we can extend the framework when the pure free states constitute a continuous set, such as in the case of fermionic non-Gaussianity~\cite{diasClassicalSimulationNonGaussian2023, reardon-smithImprovedSimulationQuantum2023, cudbyGaussianDecompositionMagic2023}.
Second, it is interesting to investigate whether it is possible to scale further computations for weakly decoupled states such as tensor network states. While exact computation may require as costly calculation as in the generic case, we may perform approximate computation with an accuracy bound.

\section*{Acknowledgments}
We thank T. Oki for his valuable comments on the manuscript, \black{and R. Takagi, B. Regula, and J. Seddon for fruitful discussions.}
N.Y. wishes to thank JST PRESTO No. JPMJPR2119 and the support from IBM Quantum. 
This work was supported by JST Grant Number JPMJPF2221. 
This work was supported by JST ERATO Grant Number  JPMJER2302 and JST CREST Grant Number JPMJCR23I4, Japan.  

\black{
\textit{Note added.---}
Hantzko et al.~\cite{hantzko2023tensorized}, and we independently proposed Pauli decomposition algorithms that share the essence of recursively slicing and transforming a density matrix. Their first preprint was uploaded about two weeks before our first preprint was uploaded. After T.~Jones informed us of their Pauli decomposition algorithm along with a Python benchmark~\cite{jones2024DensePauliDecomposer}, we worked on the \Cpp{} benchmark shown in Appendix~\ref{subapp:pauli_vec_innerproduct}.
}

\bibliographystyle{quantum}
\bibliography{robustnessOfMagic.bib,robustnessOfMagic2.bib}

\appendix

\section{Pseudocode of Fast Walsh--Hadamard Transform Algorithm} \label{app:fwht_pseudocode}
Here, we provide the pseudocode of the Fast Walsh--Hadamard Transform (FWHT) algorithm, which computes $H_n v = H_n^\top v = \mqty(1&1\\1&-1)^{\otimes n} v$ for $v\in \bbR^{2^n}$. We omitted the normalization factor in the pseudocode since the Walsh--Hadamard matrix $H_n$ itself is unnormalized.
The time complexity is $\order{n 2^n}$ and the space complexity is $\order{2^n}$, since this is a in-place algorithm.

\begin{algorithm}[ht]
    \SetKwFunction{FWHT}{FWHT}
    \SetAlgoLined
    \SetKw{KwBy}{by}
    \KwIn{$v \in \mathbb{R}^{2^n}$}
    \KwOut{In-place computation result of $H_n v$}
    \SetKwProg{Fn}{Function}{}{}
    \Fn{\FWHT{$v$}}{
        $h \leftarrow 1$ \;
        \While{$h < 2^n$}{
            \For{$i \leftarrow 0$ \KwTo $2^n-2h$ \KwBy $2h$}{
                \For{$j \leftarrow i $ \KwTo $i + h - 1$}{
                    $x \leftarrow v_j$ \;
                    $y \leftarrow v_{j + h}$ \;
                    $v_{j} \leftarrow x + y$ \;
                    $v_{j + h} \leftarrow x - y$ \;
                }
            }
            \tcp*[h]{$v \leftarrow v / \sqrt{2}$, if normalize.}\;
            $h \leftarrow 2h$ \;
        }
    }
    \caption{Fast Walsh--Hadamard Transform (FWHT) Algorithm}
    \label{alg:fwht}
\end{algorithm}

\section{Basic Properties of Pauli Vector Representation}
This section reviews the basic properties of Pauli vector representation.

\subsection{Complexity of Computing Pauli Vector Representation}\label{subapp:pauli_vec_innerproduct}
Let $\rho$ be an $n$-qubit quantum state whose Pauli vector representation is given by $b_j = \Tr[P_j \rho]$ where $P_j$ is the $j$-th Pauli operator.
To compute all the elements, naively, the computational complexity scales as $\order{8^n}$ even if we use the sparse structure of each Pauli matrix.
In the following, we show that we can perform an in-place computation that exponentially reduces the time complexity: 
\begin{lemma}[Complexity of computing Pauli vector]
Given the full density matrix representation of $n$-qubit quantum state $\rho$, its Pauli vector representation can be computed with time complexity of $\order{n 4^n}.$
\end{lemma}
\begin{proof}
First, let us introduce a map from an $n$-qubit density matrix to a $2n$-qubit state vector as follows:
\begin{equation*}
    \rho = \sum_{i, j} \rho_{i_1\cdots i_n, j_1\cdots j_n} \ketbra{i_1\cdots i_n}{j_1\cdots j_n} \mapsto
    \sum_{i,j} \rho_{i_1\cdots i_n,j_1\cdots j_n} \ket{i_1, j_1, \dots, i_n, j_n}.
\end{equation*}
Note that this is different from the well-known Choi map $\rho \mapsto \sum_{i,j} \rho_{i,j} \ket{i}\ket{j}$, and thus we refer to it as modified Choi vectorization. 
We introduce a modified Choi vector $\bm{c}$. Let $c_k$ denote the $k$-th element of $\bm{c}$, practically obtained via extracting the matrix elements of $\rho$ in the Z-order curve. 
Then, we find that $\bm{c}$ is related with the Pauli vector representation $\bm{b}$ as 
\begin{equation*}
    \bm{b} = T_n \bm{c} \quad \text{where} \quad
    T_n \defeq \mqty(
    1 & 0 &  0 &  1  \\
    0 & 1 &  1 &  0  \\
    0 & i & -i &  0  \\
    1 & 0 &  0 & -1
    )^{\otimes{n}}.
\end{equation*}
\black{
For example, for $n=1$,
we have
\begin{equation*}
    \rho=\mqty(\rho_{1,1}&\rho_{1,2}\\\rho_{2,1}&\rho_{2,2}), \quad
    \bm{c} = \mqty(\rho_{1,1}\\\rho_{1,2}\\\rho_{2,1}\\\rho_{2,2}), \quad
    \bm{b} = \mqty(\rho_{1,1}+\rho_{2,2}\\\rho_{1,2}+\rho_{2,1}\\i\rho_{1,2}-i\rho_{2,1}\\\rho_{1,1}-\rho_{2,2}) = \mqty(\Tr[\makebox[\widthof{Z}][c]{I}\rho]\\\Tr[\makebox[\widthof{Z}][c]{X}\rho]\\\Tr[\makebox[\widthof{Z}][c]{Y}\rho]\\\Tr[Z\rho]).
\end{equation*}
}
Similar to the FWHT algorithm as provided in Algorithm~\ref{alg:fwht}, in-place computation for such a tensor-product structure can be done with time complexity of $\order{n 4^n}$ and space complexity of $\order{4^n}$,
which completes the proof.
\end{proof}

\black{We describe the numerical comparison of our Pauli decomposition algorithm with other related studies~\cite{hantzko2023tensorized,jones2024decomposing,vidal2023paulicomposer} which are available via GitHub repository~\cite{paulidecomp}. 
The algorithms are implemented with \Cpp{} and run on the laptop. Our \Cpp{} implementation is mainly based on the Python implementation by Jones~\cite{jones2024DensePauliDecomposer}, except for the iterative algorithm by Hantzko et al.~since it is not included in the Python implementation. The algorithms run on 50 random density matrices; each entry's real and imaginary parts are independently and uniformly sampled from $[0, 1)$.}

\black{Fig.~\ref{fig:paulidecomp_benchmark} shows the benchmark of Pauli decomposition. It shows that our Pauli decomposition algorithm is the fastest. Although our algorithm and the algorithms by Hantzko et al.~both have the best time complexity $\order{n 4^n}$ among the algorithms, our algorithm is faster by a constant factor. We consider it is because our algorithm uses in-place computation, which is cache efficient.}

\begin{figure}[tb]
    \begin{center}
        \includegraphics[width=\columnwidth]{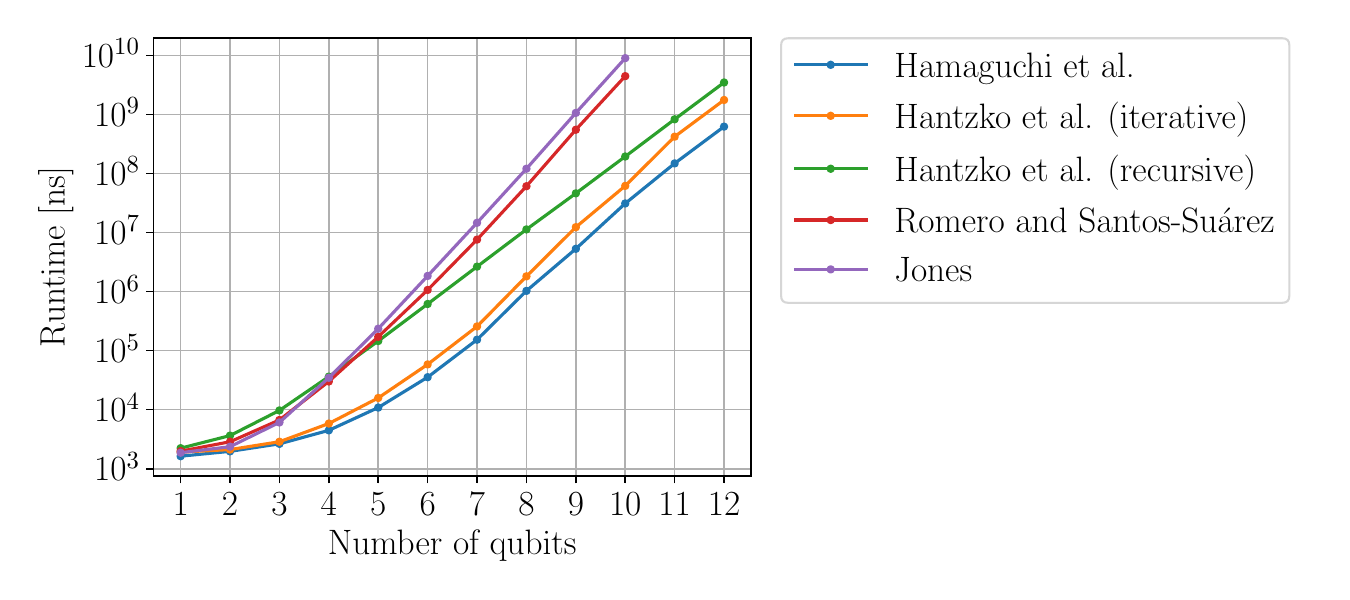}
        \caption{
        Comparison of Pauli decomposition algorithms. The algorithms are implemented with \Cpp{} and run on the laptop. The first three methods in the legend have a time complexity of $\order{n 4^n}$, while the next two methods have a time complexity of $\order{8^n}$.}
        \label{fig:paulidecomp_benchmark}
    \end{center}
\end{figure}

\subsection{Bound on Overlap Counts}
We provide a fact regarding the distribution of overlaps in addition to $0 \leq \bm{a}_j^\top \bm{b} \leq 2^n$.
\begin{lemma}[Bound on overlap counts]
    Let $\rho$ be an arbitrary $n$-qubit quantum state. Then, for all $n\in\mathbb{N}$, the count on the pure stabilizer states satisfies the following:
    \begin{equation*}
        \abs{\{ \bm{a}_j \in \calA_n \mid \bm{a}_j^\top \bm{b} \in [0,1]\}} \geq \abs{\calS_n}/2^n, \:
        \abs{\{ \bm{a}_j \in \calA_n \mid \bm{a}_j^\top \bm{b} \in [1,2^n]\}} \geq \abs{\calS_n}/2^n.
    \end{equation*} 
\end{lemma}
\begin{proof}
    Consider an arbitrary sparsified Walsh--Hadamard matrix $W_j$ in Lemma~\ref{lem:Amat_decomposition}, and let $\calW \subset \calA_n$ denote its columns. It follows that
    \begin{equation}\label{eq:overlap_tmp}
        \sum_{\bm{a} \in \calW} \bm{a}^\top \bm{b} = 2^n.
    \end{equation}
    Therefore, assuming that the overlaps are either entirely $\bm{a}^\top \bm{b} < 1$ or entirely $\bm{a}^\top \bm{b} > 1$ contradicts Eq.~\eqref{eq:overlap_tmp}. Applying the same reasoning for every $W_j$ completes the proof.
\end{proof}

\section{Proof of Proposition~\ref{prop:cover_set}} \label{app:proof_proposition_cover_set}
\black{In this section, we prove Proposition~\ref{prop:cover_set}.
\coverset*
Let us describe the structure of the subsequent sections.
Section~\ref{app:connection_cover_set_mubs} defines MUBs, discusses the connection between the cover set and MUBs, and provides a short proof by using the previous result by Gibbons et al.~\cite{GibbonsHW2004Discrete}. 
For convenience, the subsequent sections provide a full proof of Proposition~\ref{prop:cover_set}. The full proof shares the essence with Ref.~\cite{GibbonsHW2004Discrete}, but dispenses with concepts and discussions by Ref.~\cite{GibbonsHW2004Discrete} that are not useful for our goal.}

\subsection{Connection between Cover Set and MUBs}
\label{app:connection_cover_set_mubs}

\black{
MUBs are defined as bases of the state vector space that satisfy the following conditions: for any two different bases of MUBs, any state of one basis has the same amplitudes with respect to the other basis.
Gibbons et al.~\cite[Section 4]{GibbonsHW2004Discrete} constructs MUBs of maximal size $2^n + 1$ by using the following fact.
\begin{lemma}[Partition of Pauli Operators~{\cite[Section 4]{GibbonsHW2004Discrete}}]
    \label{lem:pauli_partition}
    The set of all non-identity $n$-qubit Pauli operators $\qty{I, X, Y, Z}^{\otimes n} \setminus \qty{I^{\otimes n}}$, which has $4^n - 1$ elements, can be partitioned into $2^n + 1$ subsets of $2^n - 1$ elements that commute with one another. 
\end{lemma}
Let $\calT_1, \dots, \calT_{2^n + 1}$ be the subsets of commuting Pauli operators in Lemma~\ref{lem:pauli_partition}.
Ref.~\cite{GibbonsHW2004Discrete} takes a basis $B_i$ consisting of simultaneous eigenvectors of $\calT_i$ and proves that $\qty{B_1, \dots, B_{2^n + 1}}$ is MUBs of maximal size.
Also, we can construct the cover set in Proposition~\ref{prop:cover_set} by using Lemma~\ref{lem:pauli_partition}.
\begin{proof}
    Regarding $B_i$ as a column set of the matrix $\Amat_n$, the basis $B_i$ corresponds to a sparsified Walsh--Hadamard matrix $W_i$. Then, the nonzero-row set $R(W_i)$ consists of $\calT_i$ and the first row $I^{\otimes n}$ by the definition of $B_i$. Thus, Lemma~\ref{lem:pauli_partition} implies that $\calV_n = \qty{W_1, \dots, W_{2^n + 1}}$ is the desired cover set.
\end{proof}
}

\subsection{Full Proof}
\label{app:full_proof}

\subsubsection{Proof on the Existence of Minimum Cover Stabilizer Set}
\black{Next, we directly prove Proposition~\ref{prop:cover_set}.
We utilize the bijection between Pauli operators $P \in \qty{I, X, Y, Z}^{\otimes n}$ and $2n$-dimensional row vectors $r(P) \in \bbF_2^{2n}$ of check matrices. 
First, the proof of Lemma~\ref{lem:Amat_decomposition} shows that each sparsified Walsh--Hadamard matrix $W$ corresponds to a single check matrix $C = [X_n~Z_n]$. Also, the elements of the nonzero-row set $R(W)$ can be enumerated as follows: for each $\bm{v} \in \bbF_2^{n}$, take the Pauli operator $P_{\bm{v}}$ satisfying $r(P_{\bm{v}}) = \bm{v}^{\top} C = [\bm{v}^{\top} X_n~\bm{v}^{\top} Z_n]$.
Using these observations, it suffices to find a set of check matrices $\qty{C_1, \dots, C_{2^n + 1}}$ that satisfies the following condition: for every $2n$-dimensional row vector $\bm{r} \in \bbF_2^{2n}$, there exists a check matrix $C_k$ in the set and a vector $\bm{v} \in \bbF_2^n$ that satisfies $\bm{r} = \bm{v}^{\top} C_k$.
In the following, we explicitly construct the check matrices and confirm that they satisfy the condition.}


\black{
Let us define the $2^n + 1$-st check matrix as $C_{2^n + 1} = [I~O]$.
Also, for each $k$ from $1$ to $2^n$, we take the $k$-th check matrix $C_k = [X_k~Z_k]$ as follows:}
\begin{enumerate}
    \item $Z_k$ is an $n \times n$ identity matrix.
    \item $X_k$ is given in the following lemma.
\end{enumerate}

\begin{lemma}\label{lemma:Xk}
    There exists an explicit construction for a set of symmetric matrices $\{X_k \mid X_k \in \bbF_2^{n\times n}\}_{k=1}^{2^n}$ such that the following is satisfied.
    \begin{equation}\label{eq:condition_of_Xk}
        \forall\bm{v}\in\bbF_2^n\setminus\{\bm{0}\},~\{X_k \bm{v} \}_{k=1}^{2^n} = \bbF_2^n.
    \end{equation}
\end{lemma}
\black{From the fact that $X_k$ and $Z_k$ are all symmetric matrices, it follows that the condition in Lemma~\ref{lemma:check_matrix_commutativity} is satisfied, and therefore the check matrices $C_1, \dots, C_{2^n + 1}$ are valid.}

\black{We show that the constructed check matrices indeed satisfy the above condition. Let us take any $2n$-dimensional row vector $\bm{r}$. Let $\bm{x}$ and $\bm{z}$ be the former and latter half of the row vector $\bm{r}$, respectively. If $\bm{z} = \bm{0}$, then $C_{2^n + 1} = [I~O]$ and $\bm{v} = \bm{x}$ satisfy the condition. If $\bm{z} \neq 0$, Lemma~\ref{lemma:Xk} allows us to take $k$ such that $\bm{z} X_k = \bm{x}$. Using $\bm{v} = \bm{z}$ implies $\bm{v} C_k = [\bm{z} X_k~\bm{z}] = \bm{r}$.}

\subsection{Proof of Technical Lemma~\ref{lemma:Xk}}
\black{Now the remaining work is to prove Lemma~\ref{lemma:Xk}.
We first provide the explicit construction of $\{X_k\}_{k=1}^{2^n}$, and then prove that it indeed satisfies Eq.~\eqref{eq:condition_of_Xk}.}
\subsubsection{Construction of \texorpdfstring{$X_k$}{X k}}\label{subapp:construction_Xk}

We first introduce some algebraic concepts necessary for the discussion. We denote the polynomial ring over $\Fp{2}$ by $\Fp{2}[x]$. Let $f$ be an arbitrary irreducible polynomial of degree $n$. We consider a quotient ring $\Fp{2}[x] / (f)$, where $(f)$ denotes the ideal generated by $f$. Then, $\Fp{2}[x] / (f)$ is a field because $f$ is irreducible. It is also noteworthy that $\Fp{2}[x] / (f)$ is a vector space over $\Fp{2}$ and $\qty{x^0, \dots, x^{n - 1}}$ can be taken as a basis.

In what follows, we discuss the construction of $\qty{X_k}_{k=1}^{2^n}$. We define a symmetric matrix $C(x) \in (\Fp{2}[x] / (f))^{n \times n}$ as a matrix whose $(i, j)$ entry equals $x^{i + j - 2}$. Namely, $C(x)$ can be represented as follows:
\begin{equation*}
    C(x) = \mqty(
    x^0       & x^1 & x^2    & \cdots & x^{n - 1}  \\
    x^1       & x^2 &        &        &            \\
    x^2       &     & \ddots &        & \vdots     \\
    \vdots    &     &        &        &            \\
    x^{n - 1} &     & \cdots &        & x^{2n - 2} \\
    ).
\end{equation*}
Every entry of $C(x)$ can be represented as a linear combination of a basis $\qty{x^0, \dots, x^{n - 1}}$, and we define $C_i$ be a matrix consisting of such $x^i$ coefficients. In other words, $C_i \in \Fp{2}^{n \times n}$ is defined so that $C(x) = C_0 x^0 + \dots + C_{n - 1} x^{n - 1}$ holds. Note that $C_i$ is also symmetric.

We give a concrete example below. We take $n = 3$ and $f = 1 + x + x^3$, which is irreducible. Then, $C_0, C_1, C_2$ can be derived in the following way:
\begin{align*}
    C(x)
        & = \mqty(
    x^0 & x^1                  & x^2     \\
    x^1 & x^2                  & x^3     \\
    x^2 & x^3                  & x^4     \\
    )                                    \\
        & = \mqty(
    1   & x                    & x^2     \\
    x   & x^2                  & 1 + x   \\
    x^2 & 1 + x                & x + x^2 \\
    )                                    \\
        & = \underbrace{\mqty(
    1   &                      &         \\
        &                      & 1       \\
        & 1                    &         \\
        )}_{C_0} x^0
    + \underbrace{\mqty(
        & 1                    &         \\
    1   &                      & 1       \\
        & 1                    & 1       \\
        )}_{C_1} x^1
    + \underbrace{\mqty(
        &                      & 1       \\
        & 1                    &         \\
    1   &                      & 1       \\
        )}_{C_2} x^2.
\end{align*}

Next, we consider the following set of symmetric matrices:
\begin{equation*}
    \qty{\sum_{i = 0}^{n - 1} a_i C_i \mathrel{} \middle| \mathrel{} a_i \in \Fp{2}}.
\end{equation*}
Since the elements of the set are distinct, the set has $2^n$ elements. We take this set as the set $\qty{X_k}_{k=1}^{2^n}$.

\subsubsection{Proof that \texorpdfstring{$X_k$}{X k} Satisfies Lemma~\ref{lemma:Xk}}

Next, we prove that $\qty{X_k}_{k=1}^{2^n}$ given in the previous subsection satisfies Eq.~\eqref{eq:condition_of_Xk}.

By the definition of $X_k$, one can show that Eq.~\eqref{eq:condition_of_Xk} holds if and only if $\qty{C_i \bm{v}}_{i = 0}^{n - 1}$ is linearly independent for any $\bm{v} \in \Fp{2}^n \setminus \qty{\bm{0}}$. From here, we show the linear independence of $\qty{C_i \bm{v}}_{i = 0}^{n - 1}$ by proving several lemmas.

\begin{lemma} \label{lemma:Cx_nonzero}
    For any vectors $\bm{u}, \bm{v} \in \Fp{2}^n \setminus \qty{\bm{0}}$, $\bm{u}^{\top} C(x) \bm{v} \neq 0$.
\end{lemma}
\begin{proof}
Using a vector $\bm{x} = (x^0, \dots, x^{n - 1})^{\top}$, we have $C(x) = \bm{x} \bm{x}^{\top}$. Thus, by defining two polynomials $u(x) = \bm{u}^{\top} \bm{x} = \sum_{i=0}^{n - 1} u_i x^i$ and $v(x) = \bm{v}^{\top} \bm{x} = \sum_{i=0}^{n - 1} v_i x^i$, $\bm{u}^{\top} C(x) \bm{v}$ can be represented as $u(x) v(x)$. 
Hence, it suffices to show that $u(x) v(x) \neq 0$. Because $\bm{u}$ and $\bm{v}$ are nonzero, $u(x)$ and $v(x)$ are nonzero as well. Noting that $\Fp{2} / (f)$ is a field, the product $u(x) v(x)$ is also nonzero.
\end{proof}

\begin{lemma} \label{lemma:Ci_nonzero}
    For any vectors $\bm{u}, \bm{v} \in \Fp{2}^n \setminus \qty{\bm{0}}$, there exists $i$ such that $\bm{u}^{\top} C_i \bm{v} \neq 0$.
\end{lemma}
\begin{proof}
By multiplying $C(x) = C_0 x^0 + \dots + C_{n - 1} x^{n - 1}$ by $\bm{u}$ from the left and $\bm{v}$ from the right, we obtain
\begin{equation}\label{eq:uCxv_expansion}
    \bm{u}^{\top} C(x) \bm{v} = (\bm{u}^{\top} C_0 \bm{v}) x^0 + \dots + (\bm{u}^{\top} C_{n - 1} \bm{v}) x^{n - 1}.
\end{equation}
Eq.~\eqref{eq:uCxv_expansion} expresses $\bm{u}^{\top} C(x) \bm{v}$ as a polynomial with coefficients $\bm{u}^{\top} C_{i} \bm{v} \in \Fp{2}$. Since $\bm{u}^{\top} C(x) \bm{v}$ is nonzero from Lemma~\ref{lemma:Cx_nonzero}, there exists a nonzero coefficient, i.e, $\bm{u}^{\top} C_{i} \bm{v} \neq 0$ for some $i$.
\end{proof}

\begin{lemma} \label{lemma:Ci_independent}
    For any vector $\bm{v} \in \Fp{2}^n \setminus \qty{\bm{0}}$, $\qty{C_i \bm{v}}_{i = 0}^{n - 1}$ is linearly independent.
\end{lemma}
\begin{proof}
We consider a matrix $\qty[C_0 \bm{v}, \dots, C_{n - 1} \bm{v}]$ with $C_i \bm{v}$ as the column vectors. By multiplying an arbitrary nonzero vector $\bm{u} \in \Fp{2}^n \setminus \qty{\bm{0}}$ from the left, we obtain a vector $(\bm{u}^{\top} C_0 \bm{v}, \dots, \bm{u}^{\top} C_{n - 1} \bm{v})$, which is nonzero by Lemma~\ref{lemma:Ci_nonzero}. Therefore, we can confirm that the matrix  $\qty[C_0 \bm{v}, \dots, C_{n - 1} \bm{v}]$ is non-singular, which implies the linear independence of $\qty{C_i \bm{v}}_{i = 0}^{n - 1}$.
\end{proof}

Having shown Lemma~\ref{lemma:Ci_independent}, it is also proved that $\qty{X_k}_{k=1}^{2^n}$ given in the previous subsection satisfies Eq.~\eqref{eq:condition_of_Xk}, i.e., Lemma~\ref{lemma:Xk} is proved.

\section{Proof of Proposition~\ref{prop:multiplicativity}} \label{app:proof_proposition_multi}
In this section, we prove Proposition~\ref{prop:multiplicativity}.

\OptimalityForTheRestrictedProblem*

This assures that when the stabilizer set is restricted only to the tensor product states $\bigotimes_i \Amat_{n_i} \subset \Amat_n$, the optimal solution is simply a tensor product of individual optimal solutions.
The proof can be similarly done with~\cite{jaegerDirectSumsTensor1964}.
We are considering the following primal problem:
\begin{mini*}
    {\bm{x}}{\norm{\bm{x}}_1}{\label{eq:P}}{(\mathrm{P})}
    \addConstraint
    {\qty(\bigotimes_{i=1}^{m} \bm{A}_{n_i}) \bm{x}}
    {= \bigotimes_{i=1}^{m} \bm{b}_i},
\end{mini*}
and the $i$-th primal problem is formulated as
\begin{mini*}
    {\bm{x}_i}{\norm{\bm{x}_i}_1}{\label{eq:P_i}}{(\mathrm{P}_i)}
    \addConstraint{\bm{A}_{n_i} \bm{x}_i}{= \bm{b}_i}.
\end{mini*}
The proposition asserts that if there exist optimal solutions $\bm{x}_i^*$ for every $({\rm P}_i)$, then $\bm{x}^* = \bigotimes_{i=1}^m \bm{x}_i^*$ is the optimal solution for $(\mathrm{P})$.
Let $\mathrm{OPT}$ be the optimal value for the original problem $(\mathrm{P})$.
It is clear that $\bm{x}^*$ is a feasible solution for $(\mathrm{P})$, which means $\mathrm{OPT} \leq \norm{\bm{x}^*}_1$.
The dual problem of $(\mathrm{P})$ is given as
\begin{maxi*}
    {\bm{y}}{\qty(\bigotimes_{i=1}^{m} \bm{b}_i^\top) \bm{y}}{\label{eq:D}}{(\mathrm{D})}
    \addConstraint{\norm{\qty(\bigotimes_{i=1}^{m} \bm{A}_{n_i}^\top) \bm{y}}_{\infty}}{\leq 1},
\end{maxi*}
and the $i$-th dual problem can also be formulated as
\begin{maxi*}
    {\bm{y}_i}{\bm{b}_i^\top \bm{y}_i}{\label{eq:D_i}}{(\mathrm{D}_i)}
    \addConstraint{\norm{\bm{A}_{n_i}^\top \bm{y}_i}_{\infty}}{\leq 1}.
\end{maxi*}
Now the strong duality of the problems assures that optimal solution $\bm{y}_i^*$ exists for all $(\mathrm{D}_i)$ satisfying $\bm{b}_i^\top \bm{y}_i^* = \norm{\bm{x}_i}_1$. By taking $\bm{y}^* \defeq \bigotimes_{i=1}^{m} \bm{y}_i^*$, it follows from the property of $L^\infty$ norm that
\begin{align*}
    \norm{\qty(\bigotimes_{i=1}^{m} \bm{A}_{n_i}^\top) \bm{y}^*} _{\infty} = \norm{\bigotimes_{i=1}^{m} \qty(\bm{A}_{n_i}^\top \bm{y}_i^*) } _{\infty}
    = \prod_{i=1}^{m} \norm{\bm{A}_{n_i}^\top\bm{y}_i^* } _{\infty} \leq 1,
\end{align*}
which guarantees that $\bm{y}^*$ is a feasible solution of $(\mathrm{D})$.
Thus, the optimal value of the dual problem $(\mathrm{D})$, which is $\mathrm{OPT}$ by the strong duality, is no less than
\begin{equation*}
     \qty(\bigotimes_{i=1}^{m} \bm{b}_i^\top) \bm{y}^* = \prod_{i=1}^{m} \qty(\bm{b}_i^\top\bm{y}_i^*)
    = \prod_{i=1}^{m} \norm{\bm{x}_i^*}_1 
    = \norm{\bigotimes_{i=1}^{m} \bm{x}_i^*}_1
    = \norm{\bm{x}^*}_1,
\end{equation*}
which means $\norm{\bm{x}^*}_1 \leq \mathrm{OPT}$, i.e., $\norm{\bm{x}^*}_1 = \mathrm{OPT}$. This completes the proof of Proposition~\ref{prop:multiplicativity}.

\section{Numerical Details on RoM Calculation}
In this section, we provide details on the numerical results of RoM calculation.

\subsection{More Results on Top-Overlap Method} \label{app:numerical_details}
Here, we provide the results of a numerical experiment results regarding the top-overlap method introduced in Sec.~\ref{subsec:top_K}.
We present results for Haar random mixed, pure, and tensor product states of $n=4,5,6,7$ qubit system.

As shown in Fig.~\ref{fig:wrapper_of_RoM_dot_results}, we can see that the top-overlap method significantly outperforms a naive random selection method in all cases.
The improvement becomes more evident when we compute larger systems; in particular, for $n=7$ qubit case, it suffices to take only $K=10^{-5}$ to reach near-optimal value for all three targets.

Two remarks are in order. First, we added the column set of the cover matrix in Sec.~\ref{subsec:minimal} to ensure feasibility. 
In other words, the restriction of the column set solely using the information of overlaps may lead to a rank-deficient matrix.
Second, the run time of our algorithm for a $n=6$ qubit system was 5 seconds for computing all the overlaps and 3 minutes for solving the LP. For a $n=7$ qubit system, the overlap computation consumes 15 minutes at most and 15 minutes for solving the LP.

\subsection{Assuring feasibility}\label{subsec:AssureFeasibility}
As mentioned in the previous subsection, we cannot obtain any feasible solution if the stabilizer states are not appropriately restricted. 
One of the most robust ways to ensure the feasibility for arbitrary quantum states is to utilize the cover matrix 
as mentioned before.
The second method applicable for tensor product states is to utilize the solution obtained from small-scale systems. As shown in Appendix~\ref{app:proof_proposition_multi}, we can always obtain a feasible solution by taking tensor products, which can be used as an initial solution. 
Then, one may extend the set of stabilizers to improve the quality of the solution.

\begin{figure}[hb]
\begin{center}
\includegraphics[width=\columnwidth]{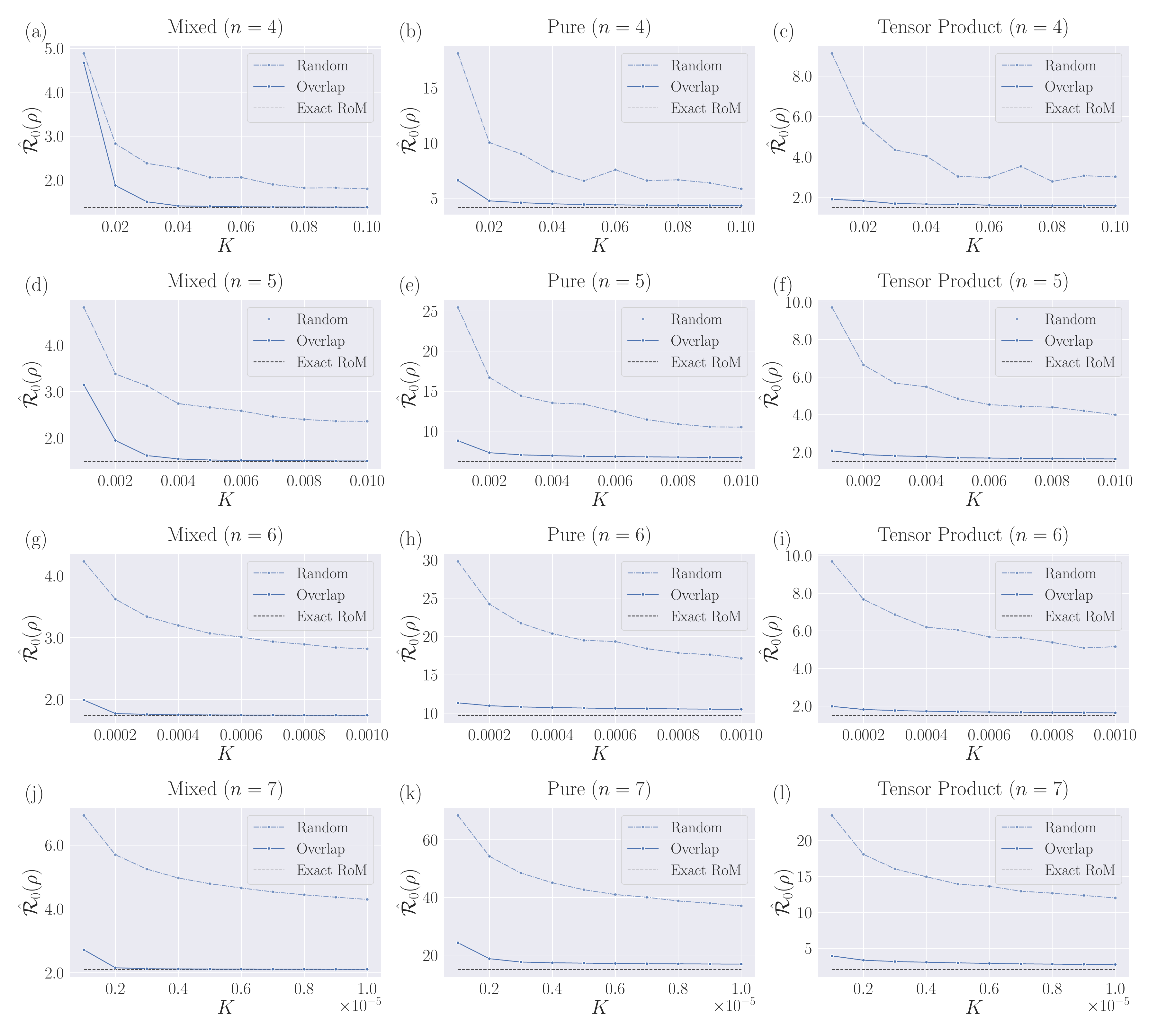}
\captionsetup{justification=centering}
    \caption{
        Numerical demonstration of top-overlap method introduced in Sec.~\ref{subsec:top_K}.
        The column set is restricted to the size of $K \abs{\calS_n}~(0 < K \leq 1)$.
        The cyan and blue lines denote the approximate RoM $\hat{\calR}_k(\rho)$ computed with a randomly restricted column set or the columns with the largest or smallest overlaps, respectively. The black dotted lines indicate the exact RoM computed with the CG method.
    }
    \label{fig:wrapper_of_RoM_dot_results}
\end{center}
\end{figure}

\end{document}